\documentclass[english,twocolumn]{IEEEtran}
\usepackage{amsfonts}
\usepackage{amsgen,amsmath,amstext,amsbsy,amsopn,amssymb}
\usepackage{graphicx}
\usepackage[usenames]{color}

\usepackage[T1]{fontenc}
\usepackage[latin9]{inputenc}
\usepackage{float}
\usepackage{bm}
\usepackage{multirow}
\usepackage{subfigure}
\makeatletter

\floatstyle{ruled}
\newfloat{algorithm}{tbp}{loa}
\floatname{algorithm}{Algorithm}

\newtheorem{definition}{Definition}

\newtheorem{lemma}{Lemma}
\newtheorem{theorem}{Theorem}

\usepackage{bm}
\usepackage{cite}
\usepackage{url}
\usepackage{babel}

\makeatother

\newcommand{\be}{\begin{equation}}
\newcommand{\ee}{\end{equation}}

\newcommand{\bbR}{\mathbb{R}}

\newcommand{\0}{{\mathbf{0}}}

\newcommand{\e}{{\mathbf{e}}}

\renewcommand{\u}{{\mathbf{u}}}
\renewcommand{\v}{{\mathbf{v}}}

\newcommand{\x}{{\mathbf{x}}}
\newcommand{\y}{{\mathbf{y}}}
\newcommand{\z}{{\mathbf{z}}}

\newcommand{\I}{{\mathbf{I}}}

\newcommand{\supp}{{\rm supp}}

\newcommand{\tx}{{\mathbf{\tilde{x}}}}
\newcommand{\tS}{{\tilde{S}}}

\newcommand{\bPhi}{\mathbf{\Phi}}

\allowdisplaybreaks

\begin{document}
\title{Subspace Thresholding Pursuit: A Reconstruction Algorithm for Compressed Sensing
}
\author{Chao-Bing Song, Shu-Tao Xia, Xin-Ji Liu}
\date{}
\maketitle
\begin{abstract}
We propose a new iterative greedy algorithm for reconstructions of sparse signals with or without noisy perturbations in compressed sensing.
The proposed algorithm, called \emph{subspace thresholding pursuit} (STP) in this paper, is a simple combination of subspace pursuit and iterative hard thresholding.
Firstly, STP  has the theoretical guarantee comparable to that of $\ell_1$ minimization in terms of restricted isometry property. Secondly, with a tuned parameter, on the one hand, when reconstructing Gaussian signals, it can outperform other state-of-the-art reconstruction algorithms greatly; on the other hand, when reconstructing constant amplitude signals with random signs, it can outperform other  state-of-the-art iterative greedy algorithms and even outperform $\ell_1$ minimization if the undersampling ratio is not very large. In addition, we propose a simple but effective method
to improve the empirical performance further if the undersampling ratio is large. Finally, it is showed that other iterative greedy algorithms can  improve their empirical performance by borrowing the idea of STP.

\end{abstract}
\renewcommand{\thefootnote}{\fnsymbol{footnote}} \footnotetext[0]{
This research is supported in part by the Major State Basic Research
Development Program of China (973 Program, 2012CB315803), the National
Natural Science Foundation of China (61371078), and the Research Fund for the Doctoral Program of Higher Education of China (20130002110051)
.

All the authors are with the Graduate School at ShenZhen, Tsinghua University, Shenzhen, Guangdong 518055, P.R. China (e-mail: scb12@mails.tsinghua.edu.cn,  xiast@sz.tsinghua.edu.cn, liuxj11@mails.tsinghua.edu.cn).} \renewcommand{\thefootnote}{\arabic{footnote}} \setcounter{footnote}{0}
 \section{\label{sec:Introduction}Introduction}

As a new paradigm for signal sampling, compressed sensing (CS) \cite{donoho2006compressed,candes2005decoding,candes2006robust}
has attracted a lot of attention in recent years. Consider an $s$-sparse signal
$\x=(x_1,x_2,\ldots,x_N)^T\in \bbR^N$ which has at most $s$ nonzero entries. Let
$\bPhi\in\bbR^{m\times N}$ be a measurement matrix with $m\ll N$ and $\y=\bPhi\x$ be
a measurement vector. CS deals with recovering the original signal $\x$ from
the measurement vector $\y$ by finding the sparsest solution to the underdetermined
linear system $\y=\bPhi\x$, i.e., solving the following \emph{$\ell_0$ minimization}
problem:
\be
\min \|\x\|_0\qquad s.t.\qquad \bPhi\x=\y, \label{eq:l_0}
\ee
where $\|\x\|_0:=|\{i:x_i\neq 0\}|$ denotes the $\ell_0$ quasi norm of $\x$.
Unfortunately, as a typical combinatorial optimization problem, the above $\ell_0$ minimization
is NP-hard \cite{candes2005decoding}. One popular strategy is to relax the $\ell_0$ minimization
problem to an \emph{$\ell_1$ minimization} problem:
\be
\min \|\x\|_1\qquad s.t.\qquad \bPhi\x=\y. \label{eq:l_1}
\ee
Due to the convex essence of $\ell_1$ minimization, we can solve it in polynomial time \cite{candes2005decoding}. However, its computational complexity is $O(m^2 N^{3/2})$ when interior point methods are employed \cite{nesterov1994interior}, which is too high for practical use.

Compared with $\ell_1$ minimization, the family of iterative greedy algorithms can reduce the computational complexity greatly, possess a similar empirical performance and have the theoretical reconstruction guarantee by the so-called \emph{restricted isometry property} (RIP). As powerful alternatives to $\ell_1$ minimization, a lot of iterative greedy algorithms  have been proposed and analyzed. According to the way of greedily selecting the columns of measurement matrix, we can divide current iterative greedy algorithms into two kinds: 1) variants of orthogonal matching pursuit (OMP) \cite{pati1993orthogonal} called OMP-like algorithms, such as OMP itself,
regularized OMP \cite{needell2010signal},  compressive sampling matching pursuit (CoSaMP) \cite{needell2009cosamp}, subspace pursuit (SP) \cite{dai2009subspace},
generalized OMP (GOMP)\cite{shim2012generalized} or orthogonal multi matching pursuit (OMMP) \cite{liu2012orthogonal}, sparsity adaptive matching pursuit (SAMP)\cite{do2008sparsity},  forward backward pursuit (FBP)\cite{karahanoglu2013compressed}; 2) variants of iterative hard thresholding (IHT) \cite{Blumensath2009} called IHT-like algorithms, such as IHT itself, gradient descent with sparsification (GDS) \cite{garg2009gradient},  hard thresholding pursuit (HTP) \cite{foucart2011hard}, normalized iterative hard thresholding (NIHT) \cite{blumensath2010normalized}. In all of these algorithms, we choose SP and CoSaMP as  repesentatives for OMP-like algorithms and HTP and NIHT as representatives for the IHT-like algorithms. They have provable  theoretical guarantees comparable to that of $\ell_1$ minimization and good empirical performance to reconstruct constant amplitude signals with random signs  (CARS signals) when compared with other iterative greedy algorithms.

In the view of theoretical guarantees, one of the most widely known conditions is the restricted isometry property (RIP) \cite{candes2005decoding} as follows.
\begin{definition}[\cite{candes2005decoding}]
\label{def:rip}
The measurement matrix $\mathbf{\Phi}\in\mathbb{R}^{m\times N}$ is said to satisfy the $s$-order RIP if for any $s$-sparse ($\|\x\|_0\le s$) signal $\mathbf{x}\in\mathbb{R}^{N}$
\begin{equation}
(1-\delta)\Vert\mathbf{x}\Vert_2^2\le\Vert\mathbf{\Phi}\mathbf{x}\Vert_2^2\le(1+\delta)\Vert\mathbf{x}\Vert_2^2, \label{eq:origin_def}
\end{equation}
where $0\le\delta\le1$. The infimum of $\delta$, denoted by $\delta_s$, is called the \emph{restricted isometry constant (RIC)} of $\mathbf{\Phi}$.
\end{definition}

Table \ref{t1} shows the sufficient conditions with respect to RICs $\delta$ with some orders for OMP, IHT and the four representatives above to perfectly reconstruct s-sparse signals.
\begin{table*}[t]
\caption{Sufficient conditions of iterative greedy algorithms to perfectly reconstruct s-sparse signals}\label{t1}
  \centering
  \begin{tabular}{|c|c|c|c|c|c|}
  \hline
  OMP \cite{mo2012remark}, \cite{wang2012recovery}& SP \cite{song2013improved} & CoSaMP \cite{song2013improved} & IHT/HTP \cite{foucart2011hard}  &NIHT \cite{blanchard2014greedy} \\
  \hline
  $\delta_{s+1}<\frac{1}{\sqrt{s}+1}$ & $\delta_{3s}<0.4859$ & $\delta_{4s}<0.5$ & $\delta_{3s}<\frac{1}{\sqrt{3}}\approx 0.5773$  & $\delta_{3s}<0.2$\\
  \hline
  \end{tabular}
\end{table*}

In the view of empirical performance,
all iterative greedy algorithms have good performance when reconstructing Gaussian signals, while relatively bad performance when reconstructing CARS signals. In \cite{maleki2010optimally}, Maleki and Donoho showed that CARS signals may be the most difficult kind of signals that iterative greedy algorithms can reconstruct. It is noteworthy that although a lot of iterative greedy algorithms can outperform $\ell_1$ minimization when reconstructing Gaussian signals, to the best of our knowledge, there is no existing iterative greedy algorithm that can outperform $\ell_1$ minimization when reconstructing CARS signals.

In this paper, our main contribution is a new algorithm, termed \emph{subspace thresholding pursuit} (STP). By finding that the idea of IHT-like algorithms can improve the approximation effect of OMP-like algorithms efficiently, we combine the steps of SP and IHT in one iteration, thus acquiring a better empirical performance. It is very convenient to  analyze  STP  theoretically since the theoretical guarantees of SP and IHT are established well and STP is only a simple combination of them.
Compared with the existing iterative greedy algorithms, with a tuned parameter $\mu$, the empirical performance of STP can be better obviously when reconstructing both Gaussian signals and CARS signals. Compared with $\ell_1$ minimization, if the undersampling ration (i.e., $m/N$) is not very large, the  empirical performance of STP  can be much better  when reconstructing Gaussian signals and can be slightly better when reconstructing CARS signals. To the best of our knowledge, it is the first iterative greedy algorithm that can  outperform the reconstruction capability of $\ell_1$ minimization in the CARS signal case. In addition, we propose a simple but effective method
to improve the empirical performance further if the undersampling ratio is large. Furthermore we
 generalize the idea of STP to other state-of-the-art iterative greedy algorithms and the resulting algorithms show better empirical performance than the original ones.

Notations: Let $\mathbf{x}\in\mathbb{R}^N$. Let $T\subseteq \{1,2,\ldots,N\}$, and $|T|$ and $\overline{T}$ respectively denote the cardinality and complement of $T$. Let $\mathbf{x}_T\in\mathbb{R}^N$ denote the vector obtained from $\mathbf{x}$ by keeping the $|T|$ entries in $T$ and setting all other entries to zero.  Let $\text{supp}(\mathbf{x})$ denote the support of $\mathbf{x}$ or the set of indices of nonzero entries in $\mathbf{x}$. Note that $\mathbf{x}$ is $s$-sparse if and only if $|\text{supp}(\mathbf{x})|\le s$. For a matrix $\mathbf{\Phi}\in\mathbb{R}^{m\times N}$, let $\mathbf{\Phi}^{\! *}$ denote the  transpose of $\mathbf{\Phi}$ and $\mathbf{\Phi}_T$ denote the submatrix that consists of columns of $\mathbf{\Phi}$ with indices in $T$. Let $\mathbf{I}$ denote the identity matrix whose dimension is decided by contexts. In addition, let $\tau=\frac{m}{N}$ denote the undersampling ratio.

Denote the general CS model:
\begin{equation}
\mathbf{y}=\mathbf{\Phi}\mathbf{x}+\mathbf{e}=\mathbf{\Phi}\mathbf{x}_S+\mathbf{\Phi}\mathbf{x}_{\overline{S}}+\mathbf{e}=\mathbf{\Phi}\mathbf{x}_S+\mathbf{e}^{\prime},
\label{eq:general_model}
\end{equation}
where $\mathbf{\Phi}\in\mathbb{R}^{m\times N}$ is a measurement matrix with $m\ll N$, $\mathbf{e}\in\mathbb{R}^m$
is an arbitrary noise, $\mathbf{y}\in\mathbb{R}^m$ is a low-dimensional observation, and $\mathbf{e}^{\prime}=\mathbf{\Phi}\mathbf{x}_{\overline{S}}+\mathbf{e}$
denotes the total perturbation by the sparsity defect $\mathbf{x}_{\overline{S}}$ and measurement error $\mathbf{e}$.

The remainder of the paper is organized as follows. Section \ref{sec:alg} gives the algorithm description and  theoretical analyses. Section \ref{sec:simulations} gives the performance simulations and analyses. Section \ref{sec:extended-study} gives some extended study. Finally, we conclude the paper in Section \ref{sec:conclusion}.


\section{\label{sec:alg}Algorithm Description and Theoretical Analyses}
In this section, we give the  description of the STP algorithm in subsection \ref{sub:alg-description}. Then the theoretical guarantee and its proof are given in subsection \ref{sub:guarantee}. Finally, we show the upper bound of the number of the STP's iterations. The proofs of the theoretical guarantee and the number of iterations can be found in appendix.
\subsection{\label{sub:alg-description}Algorithm Description}
Firstly, in order to give a clear comparison of SP, HTP and the proposed STP, the SP and HTP algorithms are summarized as follows.

\begin{algorithm}[H]
Input: $\y,\bPhi,s$.\\
Initialization: $S^{0}=\emptyset,\x^0=\0$.\\
Iteration: At the $n$-th iteration, go through the following steps.
\begin{enumerate}
\item $\Delta S =$ \{$s$ indices corresponding to the  $s$ largest magnitude entries in the vector $\bPhi^{\! *}\, (\y- \bPhi \x^{n-1})$\}.
\item $\tS^{n}=S^{n-1} \bigcup \Delta S$.
\item $\tx^{n}=\text{arg}\min_{\z \in\bbR^N}\{ \Vert\y-\bPhi\z\Vert_2,\;
    \text{supp}(\z)\subseteq \tS^{n}\}$.
\item $S^{n}=$\{$s$ indices corresponding to the $s$ largest magnitude elements of $\mathbf{\tilde{x}}^{n}$\}.
\item $\x^{n}=\text{arg}\min_{\z \in\bbR^N}\{ \Vert\y-\bPhi\z\Vert_2,\;
    \text{supp}(\z)\subseteq S^{n}\}$.
\end{enumerate}
until the stopping criteria is met. \\
Output: $\mathbf{x}^{n}$, \text{supp}($\mathbf{x}^{n}$).
\caption{Subspace Pursuit}\label{alg:sp}
\end{algorithm}

\begin{algorithm}[H]
Input: $\y,\bPhi,s$.\\
Initialization: $S^{0}=\emptyset,\x^0=\0$.\\
Iteration: At the $n$-th iteration, go through the following steps.
\begin{enumerate}
\item $S^{n}=$\{$s$ indices correspoding to the $s$ largest magnitude entries of $\x^{n-1}+\bPhi^{\! *}(\y-\bPhi\x^{n-1})$\}.
\item $\x^{n}=\text{arg}\min_{\z \in\bbR^N}\{ \Vert\y-\bPhi\z\Vert_2,\;
    \text{supp}(\z)\subseteq S^{n}\}$.
\end{enumerate}
until the stopping criteria is met. \\
Output: $\mathbf{x}^{n}$, \text{supp}($\mathbf{x}^{n}$).
\caption{Hard Thresholding Pursuit}\label{alg:htp}
\end{algorithm}

The main steps of STP  are summarized below.
\begin{algorithm}[H]
Input: $\y,\bPhi,s, \mu$.\\
Initialization: $S^{0}=\emptyset,\x^0=\0$.\\
Iteration: At the $n$-th iteration, go through the following steps.
\begin{enumerate}
\item $\Delta S =$ \{$s$ indices corresponding to the  $s$ largest magnitude entries in the vector $\bPhi^{\! *}\, (\y- \bPhi \x^{n-1})$\}.
\item $\tS^{n}=S^{n-1} \bigcup \Delta S$.
\item $\tx^{n}=\text{arg}\min_{\z \in\bbR^N}\{ \Vert\y-\bPhi\z\Vert_2,\;
    \text{supp}(\z)\subseteq \tS^{n}\}$.

\item $U^{n}=$\{$s$ indices corresponding to the $s$ largest magnitude elements of $\mathbf{\tilde{x}}^{n}$\}.
\item $\mathbf{u}^{n}=$ \{the vector from ${\mathbf{\tilde{x}}}^{n}$ that keeps the entries of ${\mathbf{\tilde{x}}}^{n}$ in $U^{n}$ and set all other ones to zero.\}
\item $S^{n}=$\{$s$ indices correspoding to the $s$ largest magnitude entries of $\u^{n}+\mu\bPhi^{\! *}(\y-\bPhi\u^{n})$\}.
\item $\x^{n}=\text{arg}\min_{\z \in\bbR^N}\{ \Vert\y-\bPhi\z\Vert_2,\;
    \text{supp}(\z)\subseteq S^{n}\}$.

\end{enumerate}
until the stopping criteria is met. \\
Output: $\mathbf{x}^{n}$, \text{supp}($\mathbf{x}^{n}$).
\caption{Subspace Thresholding Pursuit}\label{alg:stp}
\end{algorithm}

The STP algorithm is initialized with a trivial signal approximation $\x^0=\0$ and a trivial support estimate $S^0=\emptyset$. The parameters $\mu$ can be adjusted before the execution of STP. In each iteration, we call steps 1 and 2 ``OMP-like identification'' since they are common identification steps for all  OMP-like algorithms.
 Such identification steps select the set $\Delta S$ of the indices corresponding to the  one or several largest entries in $\bPhi^{\! *}\, (\y- \bPhi \x^{n-1})$  and then merge $\Delta S$ and the support estimate $S^{n-1}$. Then in step 3, STP solves a least squares problem to approximate the original signal $\x$ on the merged set $\tilde{S}^{n}$. In steps 4 and 5, STP employs a pruning stage by retaining only the $s$ largest entries in the least squares signal approximation $\tx$ to produce a new approximation $\u^{n}$. Step 6 is a common step for all  IHT-like algorithms which we call ``IHT-like identification''. In the IHT-like identification step,
 STP selects the set $S^n$ of indices corresponding to the $s$ largest entries in the vector $\u^{n}+\mu\bPhi^{\! *}(\y-\bPhi\u^{n})$. Finally, a least squares problem is solved again to get the final approximation $\x^{n}$ in the $n$-th iteration.

The stopping criteria of iterative greedy algorithms can be selected differently in implementation. One alternative is to use the stopping criteria according to the  property of the corresponding algorithm,
such as ``$n> s$'' of OMP, ``$\|\y-\bPhi\x^n\|_2\ge\|\y-\bPhi\x^{n-1}\|_2$'' of SP in \cite{dai2009subspace} or ``$S^{n-1}=S^{n}$'' of HTP in \cite{foucart2011hard},
or the stopping criteria that is independent from the algorithm itself can be used, which may be ``$n> n_{\max}$ or $\|\y-\bPhi\x^{n}\|_2<\varepsilon \|\y\|_2 $''.
If an algorithm is stable, i.e., as the iteration process continues, the series $\{\|\x^{n}-\x_S\|_2, n=1,2,3,\cdots\}$ will not diverge, such a criteria provides a tradeoff between accuracy and computational complexity.

OMP-like algorithms and IHT-like algorithms have big differences in the identification steps, but
the representatives SP and HTP of both kinds respectively have nearly the same empirical performance. The main characteristic of
the proposed STP algorithm is that we combine the ``OMP-like identification'' and ``IHT-like identification'' in one iteration,
so as to take full advantage of the virtue of both SP and HTP.
In the view of IHT-like algorithms, supp($\x$) can lie both in the support of $\u^n$ and $\x-\u^n$, where $\bPhi^{\! *}(\y-\bPhi\u^n)=\bPhi^{\! *}\bPhi(\x_S-\u^n)+\bPhi^{\! *}\e^{\prime}$ is a one-dimensional approximation of $\x-\u^n$ if the RIC is small ($\bPhi^{\! *}\bPhi\approx \I$ in this case).   But the pruning process is only taken in  $\tilde{S}^n$. Therefore, after the pruning stage in steps 4 and 5,
taking a IHT-like identification step may be a good way to give the support of $\x$ in $S\backslash U^n$ the opportunity to enter
into the final support estimate $S^n$, so as to get a better approximation effect. The weight parameter $\mu$ is selected according
to  experience. If $\mu$ is large, STP is prone to select the indices of $\supp(\x-\u^n)$; conversely, the indices of
$\supp(\u^n)$ is preferred to be selected.

\subsection{\label{sub:guarantee}The Theoretical Guarantee}
In this subsection, we establish the sufficient condition to guarantee STP to converge in the realistic case, i.e., considering the general CS model  $\mathbf{y}=\mathbf{\Phi}\mathbf{x}_S+\mathbf{e}^{\prime}$ in (\ref{eq:general_model}) directly. The conclusion in the realistic case can be  specialized to the idealized case simply by setting $\e^{\prime}=\0$. In fact, setting $\e^{\prime}=\0$ inside our derivations would simplify them considerably.
\begin{theorem}
\label{thm:main}
 For the general CS model $\mathbf{y}=\mathbf{\Phi}\mathbf{x}_S+\mathbf{e}^{\prime}$ in (\ref{eq:general_model}), if $\dfrac{1}{1-\delta_{3s}}-\dfrac{1+\delta_{3s}}{2\delta_{3s}\sqrt{1+2\delta_{3s}^2}}<\mu<1$, or $1<
 \mu<\dfrac{1}{1+\delta_{3s}}+\dfrac{1-\delta_{3s}}{2\delta_{3s}\sqrt{1+2\delta_{3s}^2}}$, or $\mu=1$ with $\delta_{3s}<0.5340$, then the sequence of $\mathbf{x}^{n}$ defined by STP satisfies
\begin{equation}
\Vert\mathbf{x}_S-\mathbf{x}^{n}\Vert_2\le\rho^{n}\Vert\mathbf{x}_S\Vert_2
+\tau\Vert\mathbf{e}^{\prime}\Vert_2, \label{eq:main}
\end{equation}
where
\begin{eqnarray}
\rho&=&\frac{2\delta_{3s}(|\mu-1|+\mu\delta_{3s})\sqrt{1+2\delta_{3s}^2}}{1-\delta_{3s}^2}<1,\label{eq:rho} \\
(1-\rho)\tau&=&\left(\frac{\sqrt{2+\sqrt{2}}\delta_{3s}(|\mu-1|+\mu\delta_{3s})}{\sqrt{1-\delta_{3s}^2}}+1\right)\nonumber\\
&&\cdot\frac{\sqrt{2(1-\delta_{3s})}+\sqrt{1+\delta_{3s}}}{1-\delta_{3s}}\nonumber\\
&&\qquad+\frac{\sqrt{4+\sqrt{2}}(|\mu-1|+\mu\delta_{3s})}{\sqrt{1-\delta_{3s}}}.\label{eq:tau}
\end{eqnarray}
\end{theorem}

The proof of Theorem \ref{thm:main} can be found in Appendix \ref{proof1}.

In Fig. \ref{fg:delta3s}, we plot the upper bound of $\delta_{3s}$ as $\mu$ changes. From Fig. \ref{fg:delta3s}, we know that when $\mu=1$, we get the best theoretical guarantee $\delta_{3s}<0.5340$ of STP which is a little weaker than the best theoretical guarantee $\delta_{3s}<\frac{1}{\sqrt{3}}\approx0.5773$ of HTP and IHT for the kind of iterative greedy algorithms so far. In addition, when $\mu\neq1$, we can still get a theoretical guarantee comparable to that of $\ell_1$ minimization in terms of RIP, i.e., the RIC is bounded by a positive constant. However, as we will see in Section \ref{sec:simulations}, the optimal $\mu$ which makes STP attain the optimal empirical performance is the one that is larger than $1$ in most cases.


\begin{figure}
  \centering
  \includegraphics[scale=0.5]{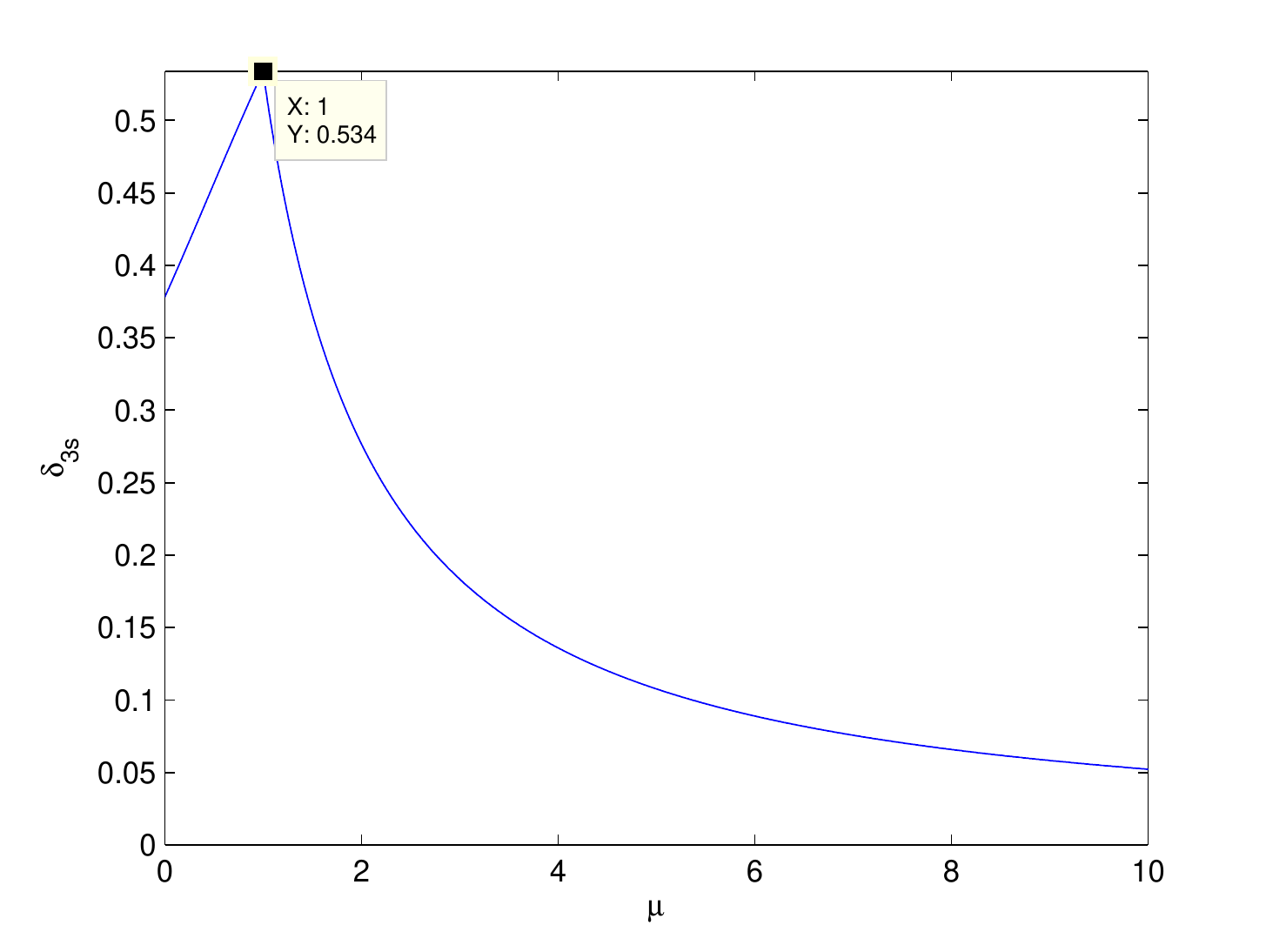}\\
  \caption{The upper bound of $\delta_{3s}$ as $\mu$ changes}\label{fg:delta3s}
\end{figure}

\subsection{\label{sub:iterations}The number of iterations}

In the following, the number of iterations of STP to reconstruct s-sparse signals in the noiseless case is considered and it is shown that STP converges in a finite number of iterations.
\begin{theorem}
\label{thm:iter-num}
Assume that the measurement matrix $\bPhi$ and the parameter $\mu$ satisfy $\rho=\frac{2\delta_{3s}(|\mu-1|+\mu\delta_{3s})\sqrt{1+2\delta_{3s}^2}}{1-\delta_{3s}^2}<1$,
particularly when $\mu=1$, $\delta_{3s}<0.5340$, then any $s$-sparse vector $\x\in\bbR^{N}$ is reconstructed by STP with $\y=\bPhi\x$ in at most
\be
\min\left\{\left\lceil \frac{\ln\|\x\|_2/\xi}{\ln {1}/{\rho}} \right\rceil, \left\lceil\dfrac{1.5s}{\ln {1}/{\rho}}\right\rceil\right\} \;\mbox{iterations}. \label{eq:iter}
\ee
where $\xi$ is defined as the smallest magnitude of all the nonzero entries in $\x$.
\end{theorem}

The proof of Theorem \ref{thm:iter-num} can be found in Appendix \ref{proof2}.

Compared with SP, STP adds steps $5$ and $6$. The runtime of step 5 is negligible and the computational complexity of step  6 in STP is comparable to step 1, i.e., the OMP-identification step, so the computational complexity analysis for SP in \cite{dai2009subspace} is also suitable for STP. Generally, STP has a comparable computational complexity with SP. In each iteration, STP needs an extra computational step, but it has less number of iterations than SP since it has a better convergence rate $\rho$ than SP under the same measurement matrix.

%
\section{\label{sec:simulations}Simulations and Discussions}

The theoretical guarantee in terms of RIC provides us the intuition of the worst-case
 performance of a reconstruction algorithm. But it does not say much about empirical performance since the theoretical guarantee by RIP is very weak (see the ``strong
 phase transition curve'' in \cite{blanchard2011phase}) and there is no efficient way to verify the RIC condition.
 In our simulations, we use the testing strategy in
 \cite{candes2005error,dai2009subspace} which measures the effectiveness of reconstruction  algorithms by checking the exact reconstruction rate in the  noiseless case. By comparing the maximal sparsity level of the underlying sparse signals at which the perfect reconstruction is ensured (this point is often
 called critical sparsity \cite{dai2009subspace}), the performance of the reconstruction can be compared empirically. In our simulations, we consider OMP, SP, CoSaMP,  NIHT, HTP,  $\ell_1$
 minimization  and STP with different $\mu$ ($\mu=1,1.5,2,2.5,3,3.5$). We let OMP execute $s$ steps. For other greedy algorithms, we use a common stopping criteria ``$n>200$ or $\|\y-\bPhi\x^n\|_2<10^{-10}\|\y\|_2$''.
For $\ell_1$ minimization, we use the default setting in the $\ell_1$-magic package (\url{http://users.ece.gatech.edu/~justin/l1magic/}).
Generally speaking, in the realistic case, the general case may be $m\ll N$, thus selecting a relatively small undersampling ratio such as $\tau=\frac{m}{N}=0.1$ may give us a better intuition about empirical performance. Therefore, in each trial, we construct a $m\times N(m=100,N=1000)$ measurement matrix $\bPhi$ with entries drawn independently from Gaussian distribution $\mathbb{N}(0,\frac{1}{m})$.
In addition, we generate an $s$-sparse vector $\x$ whose support is chosen at random. Two types of sparse signals are considered: Gaussian signals and CARS signals.
 Each nonzero element of Gaussian signals is drawn from standard Gaussian distribution and that of CARS signals is  from the set $\{1,-1\}$ uniformly at random.
  For each reconstruction algorithm, we perform 2,000 independent trials and plot the exact reconstruction rate in $y$-axis as the sparsity $s$ changes in $x$-axis.

In each figure, we plot three curves of STP with three different $\mu$: $\mu=1$, $\mu=\mu_{*}$, $\mu=\mu_{*}+0.5$, where $\mu_{*}$ stands for the $\mu$ with which  STP has the optimal empirical performance.
In Fig. \ref{fg:rate-test:a}, we show that in the Gaussian signal case,
the critical sparsity of STP with  $\mu_{*}$ (in this case $\mu_{*}=3$) exceeds all the other algorithms greatly. In Fig. \ref{fg:rate-test:b},
we show that in the CARS signal case, the empirical performance of STP with $\mu_{*}$ (in this case $\mu_{*}=2.5$) exceeds that of all the other algorithms, even including $\ell_1$ minimization.
To the best of our knowledge, all the existing iterative greedy algorithms  perform worse than $\ell_1$ minimization
in the CARS signal case. However STP with  $\mu_{*}$ breaks the limitation if the undersampling ratio is not very large, such as $\tau=0.1$.

 \begin{figure}[h]
\centering
\subfigure[Exact reconstruction rate in the Gaussian signal case]{
\label{fg:rate-test:a}
\centering
\includegraphics[scale=0.5]{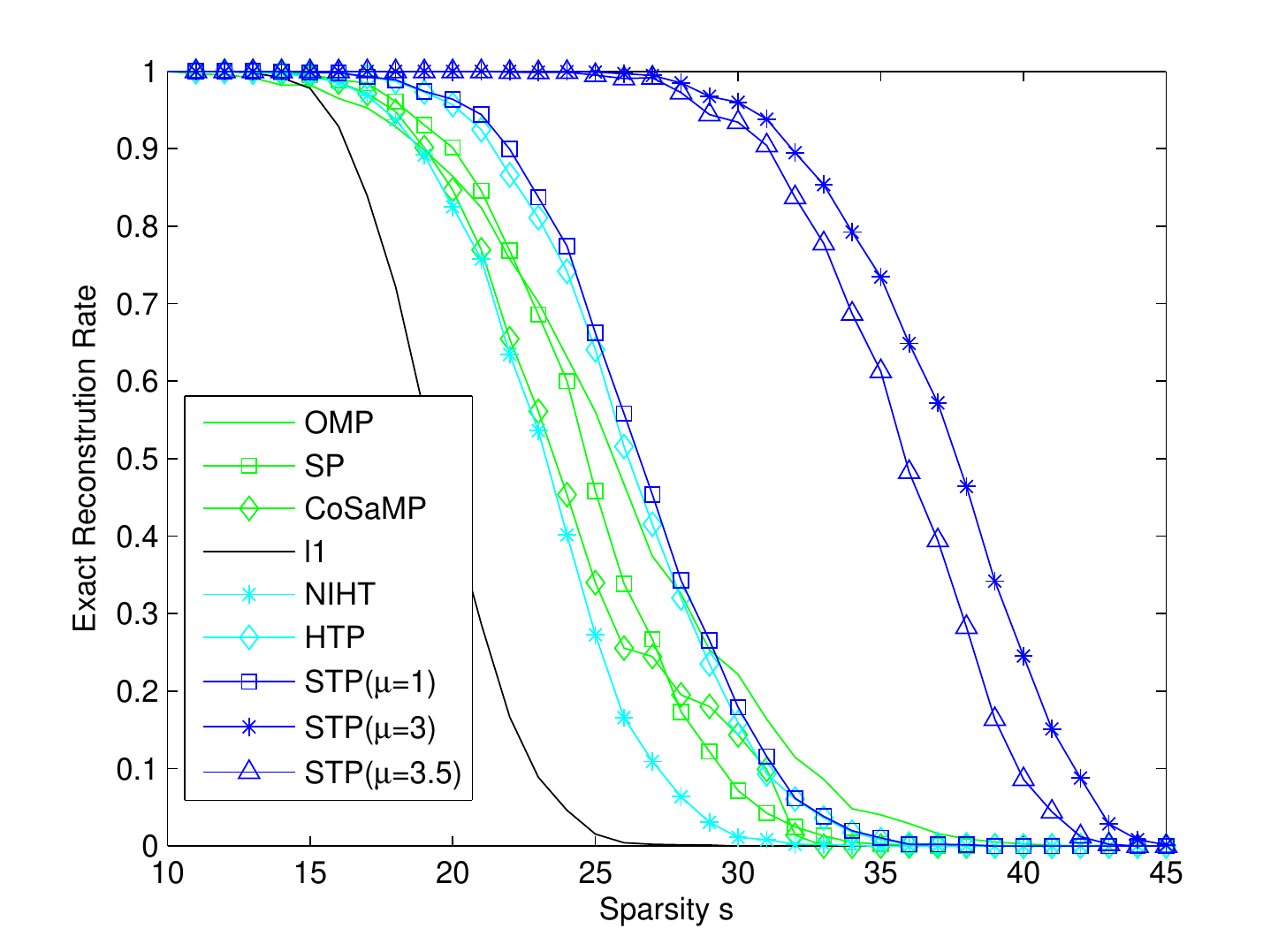}
}
\subfigure[Exact reconstruction rate in the CARS signal case]{
\label{fg:rate-test:b}
\centering
\includegraphics[scale=0.5]{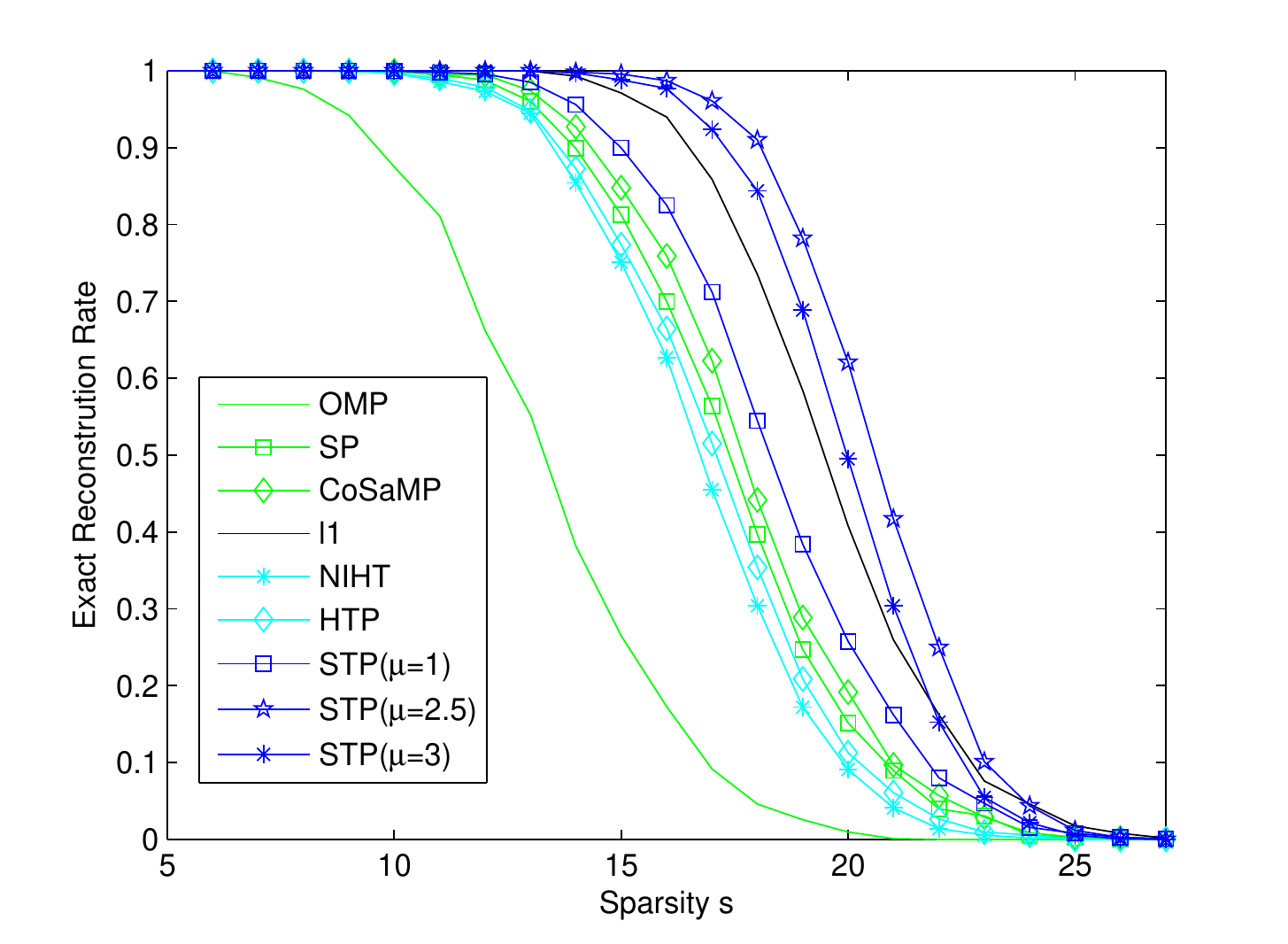}
}
\caption{Exact reconstruction rate under $100\times 1000$ Gaussian measurement matrix with $\tau=0.1$.}\label{fg:rate-test}
\end{figure}

Furthermore,
STP has its best theoretical guarantee $\delta_{3s}<0.5340$ when $\mu=1$. But from Figs. \ref{fg:rate-test:a} and \ref{fg:rate-test:b}, STP has the optimal empirical performance with $\mu_{*}=3$ in the Gaussian signal case and $\mu_{*}=2.5$ in the CARS signal case respectively.
A similar case also happens in the variant NIHT of IHT. NIHT has a better empirical performance than IHT, but its theoretical guarantee is worse than that of IHT.
On the one hand, the RIC upper bounds for different algorithms are limited by derivation skills; on the other hand, the RIC condition is only the sufficient condition
to guarantee sparse recovery, even if $\delta_{2s}$ approximates 1, we still can't say that our algorithm can't reconstruct $s$-sparse signals uniformly under some
specified measurement matrix (see the example in \cite{davies2009restricted}). Therefore, theoretical guarantee tells us how bad a algorithm will not be, but the more
important measure in practice may be empirical performance.

\begin{table*}[t]
\caption{The critical sparsity of algorithms under Gaussian measurement matrices with different sizes in the Gaussian signal case.}\label{tb:1}
\centering
\begin{tabular}{|c|c|c|c|c|c|c||c|c|c|c|c|c|c|}
\hline
Algorithms & OMP & $\ell_1$ & SP & CoSaMP & NIHT & HTP & STP1 & STP1.5 & STP2 & STP2.5 & STP3 & STP3.5 \\
\hline
Critical Sparsity(100$\times$3000) & 8 & 9 & 10 & 9 & 10 & \textbf{12} & 11 & 14 & 16 & 19   & \textbf{\emph{19}} & 16 \\
\hline
Critical Sparsity(100$\times$1000) & 10 & 12 & 13 & 13 & 12 & \textbf{14} & 15 & 18 & 22 & 24   & \textbf{\emph{24}} & 22 \\
\hline
Critical Sparsity(200$\times$1000) & 24 & 36 & \textbf{46} & 42 & 38 & 44 & 50 & 56 & 58 & 64   & \textbf{\emph{68}} & 64 \\
\hline
Critical Sparsity(300$\times$1000) & 47 & 68 & \textbf{83} & 80 & 68 & 83 & 89 & 98 & 107 & 116   & \textbf{\emph{122}} & 116\\
\hline
Critical Sparsity(400$\times$1000) & 57 & 107 & \textbf{127} & 127 & 107 & 122 & 137 & 152 & 167 & \textbf{\emph{172}}   & 172 & 162\\
\hline
Critical Sparsity(500$\times$1000) & 72 &157 &\textbf{197}   & 162 & 152 & 177 &202  &217  &\textbf{\emph{232}}  & 227   & 222 & 217 \\
\hline
Critical Sparsity(600$\times$1000) & 87 & 222 & \textbf{287} & 217 & 217 & 237 & \textbf{\emph{288}} & 287 & 287 & 282   & 277 & 267 \\
\hline
\end{tabular}
\end{table*}

\begin{table*}[t]
\caption{The critical sparsity of algorithms under Gaussian measurement matrices with different sizes in the CARS signal case.}\label{tb:2}
\centering
\begin{tabular}{|c|c|c|c|c|c|c||c|c|c|c|c|c|c|}
\hline
Algorithms & OMP & $\ell_1$ & SP & CoSaMP & NIHT & HTP & STP1 & STP1.5 & STP2 & STP2.5 & STP3 & STP3.5 \\
\hline
Critical Sparsity(100$\times$3000) & 4 & \textbf{9} & 5 & 6 & 7 & 6 & 7 & 8 & 10 & \textbf{\emph{11}} & 10 & 9 \\
\hline
Critical Sparsity(100$\times$1000) & 5 & \textbf{13} & 10 & 9 & 9 & 8 & 11 & 11 & 13 & \textbf{\emph{14}} & 14 & 13 \\
\hline
Critical Sparsity(200$\times$1000) & 12 & \textbf{38} & 31 & 31 & 29 & 29 & 31 & 36 & 36 & \textbf{\emph{37}}   & 37 & 35\\
\hline
Critical Sparsity(300$\times$1000) & 18 & \textbf{68} & 58 & 58 & 52 & 52 & 60 & 64 & \textbf{\emph{68}} & 68   & 64 & 60\\
\hline
Critical Sparsity(400$\times$1000) & 23 & \textbf{110} & 95 & 95 & 71 & 80 & 98 & \textbf{\emph{104}} & 104 & 104   & 104 & 98\\
\hline
Critical Sparsity(500$\times$1000) & 30 & \textbf{166} & 138 & 130 & 102 & 118 & 142 & \textbf{\emph{150}} & 142 & 142   & 142 & 134\\
\hline
Critical Sparsity(600$\times$1000) & 42 & \textbf{234} & 190 & 182 & 142 & 150 & 194 & \textbf{\emph{198}} & 194 & 190   & 190 & 186\\
\hline
\end{tabular}
\end{table*}

Then we focus on the the key performance measure: critical sparsity in the exact reconstruction rate curve. Tables \ref{tb:1} and \ref{tb:2} show the critical sparsity of all the algorithms in our simulations
under Gaussian measurement matrices with different sizes in the Gaussian signal case and CARS signal case respectively. In both tables, STP$c$ stands for  STP  with
$\mu=c$ and in the first columns of both tables, the formulae $m\times N(m=100,\cdots, 600, N=1000,3000$) in brackets denote the sizes of the Gaussian measurement
 matrices we used. In each row, there are two boldfaced numbers, the normal one stands for the maximal critical sparsity among the existing algorithms, the itatic one the maximal critical
 sparsity among all the STP algorithms with different $\mu$. When there exist two or more STP algorithms with different $\mu$ having an identical critical sparsity, we highlight
 the critical sparsity of the STP algorithm which has better exact reconstruction rate when the signal sparsity is larger than critical sparsity.

Firstly, from Tables \ref{tb:1} and \ref{tb:2}, we find that $\mu_{*}$ that makes STP perform best is bounded in a limited range; meanwhile, in our experiments, we find that the empirical performance of STP changes distinctly (we say algorithm A outperforms algorithm B distinctly if the exact reconstruction rate of A in any point is equal or greater than that of B) only when the step size of $\mu$ is larger than a positive value (in our case the positive value may roughly be 0.5), so in practice tuning the value of $\mu$ is just to select $\mu{_*}$ from a set of discrete values such as $\{0,0.5,1,1.5,2,2.5,3,3.5\}$.
In fact, from the two tables, in the Gaussian signal case, if $\tau\le0.3$ and the step size of $\mu$ is $0.5$, $\mu_*$ is $3$; in the
CARS signal case, if $\tau\le0.2$ and the step size of $\mu$ is $0.5$, $\mu_*$ is $2.5$. Secondly, we notice that in all the cases
($\tau=0.033, 0.1, 0.2, 0.3, 0.4, 0.5, 0.6$), STP with $\mu_{*}$ outperforms all the other greedy algorithms  when reconstructing both Gaussian signals and CARS signals.
When compared with $\ell_1$ minimization, STP with  $\mu_*$ performs worse than $\ell_1$ minimization  obviously  only when $\tau\ge0.4$ and the signal type is CARS.
In other cases, STP with $\mu_*$ has a better empirical performance, particularly the critical sparisity of STP is nearly twice than that of $\ell_1$ minimization if $\tau\le0.4$ and the signal type is Gaussian.
Thirdly, we notice that $\mu_*$ is different when reconstructing Gaussian signals and CARS signals. But since  the empirical performance of STP changes gradually as $\mu$ changes, we can select a $\mu$ between the $\mu_*$ in the Gaussian signal case and the $\mu_*$ in the CARS signal case in practice.

Generally, by our simulations, STP with $\mu_*$ has good empirical performance.
Compared with all other reconstruction algorithms, in terms of critical sparsity, it is more appropriate to realistic applications---STP will has more obvious superority as the undersampling ratio decreases.

\section{\label{sec:extended-study}Some Extended Study}
 Firstly, the empirical performance of an iterative greedy algorithm depends not only on its greedy strategy, but also on  that every step does it work rightly.
 In subsection \ref{sub:overfitting}, we propose a simple but effective way to make step 3 of STP work rightly if the undersampling ratio is large and the original signal is a Gaussian one with a sparsity around $\lceil m/2\rceil$. Secondly,
from the above sections, we know that STP is only a simple combination of SP and IHT. A direct idea may be to combine other similar iterative greedy algorithms and IHT. In subsection \ref{sub:iht-like}, we show the improvement by the IHT-like identification step to four other iterative greedy algorithms: CoSaMP, HTP, SAMP and FBP.

\subsection{\label{sub:overfitting}A method to improve the critical sparisity further }
In the description Alg. \ref{alg:stp} of STP, we have a premise that the condition number of the submatrix $\bPhi_{\tilde{S}^n}^{\! *}\bPhi_{\tilde{S}^n}$ should not be too large or infinite (the RIP condition is to say the condition number of such a submatrix should not be too large essentially), otherwise the solution to the least squares problem will be impacted by the columns in $\bPhi_{\tilde{S}^n\backslash T}$ greatly  or  not be  unique respectively. As a result, the indices in $U^n$ may be far away from the support $S$ of $\x$ in each iteration. In this case, $S^n$ will not approximate $S$ rightly as the iteration goes. This case may happen if the undersampling ratio is large and the original signal is a Gaussian one and has a sparsity around $\lceil m/2\rceil$. In Table \ref{tb:1}, we show that as the undersampling ratio increases, the critical sparsity will approximate $\lceil m/2\rceil$ gradually and in the last row of Table \ref{tb:1}, the critical sparsity of STP with $\mu_*=1$ is $288$. But in our observations, when  the sparsity $s$ is a little larger than the critical sparsity $288$, the exact reconstruction rate will decrease to zero steeply. A direct interpretation may be that our greedy strategy, i.e., combing OMP-like identification and IHT-like identification together, can do more, but as the cardinality of $\tilde{S}^n$ increases, the condition number of $\bPhi_{\tilde{S}^n}^{\! *}\bPhi_{\tilde{S}^n}$ is too large to prevent the critical sparsity increasing further.

 %

In order to address this problem, a useful way may be to reduce the number of indices selected in step 1 when the sparsity $s$ exceeds the critical sparsity. Therefore, we add some simple logic before iteration and modify STP as follows.

\begin{algorithm}[H]
Input: $\y,\bPhi,s,\mu,\gamma$.\\
Initialization: $S^{0}=\emptyset,\x^0=\0$.\\
\textbf{If} $s>\gamma m$ \\
\textbf{then} $s^{\prime}=\lceil 2\gamma m\rceil-s$.\\
\textbf{else} $s^{\prime}=s$. \\
Iteration: At the $n$-th iteration, go through the following steps.
\begin{enumerate}
\item $\Delta S =$ \{$s^{\prime}$ indices corresponding to the  $s^{\prime}$ largest magnitude entries in the vector $\bPhi^{\! *}\, (\y- \bPhi \x^{n-1})$\}.
\item $\tS^{n}=S^{n-1} \bigcup \Delta S$.
\item $\tx^{n}=\text{arg}\min_{\z \in\bbR^N}\{ \Vert\y-\bPhi\z\Vert_2,\;
    \text{supp}(\z)\subseteq \tS^{n}\}$.

\item $U^{n}=$\{$s$ indices corresponding to the $s$ largest magnitude elements of $\mathbf{\tilde{x}}^{n}$\}.
\item $\mathbf{u}^{n}=$ \{the vector from ${\mathbf{\tilde{x}}}^{n}$ that keeps the entries of ${\mathbf{\tilde{x}}}^{n}$ in $U^{n}$ and set all other ones to zero.\}
\item $S^{n}=$\{$s$ indices correspoding to the $s$ largest magnitude entries of $\u^{n}+\mu\bPhi^{\! *}(\y-\bPhi\u^{n})$\}.
\item $\x^{n}=\text{arg}\min_{\z \in\bbR^N}\{ \Vert\y-\bPhi\z\Vert_2,\;
    \text{supp}(\z)\subseteq S^{n}\}$.

\end{enumerate}
until the stopping criteria is met. \\
Output: $\mathbf{x}^{n}$, \text{supp}($\mathbf{x}^{n}$).
\caption{Subspace Thresholding Pursuit version 2}\label{alg:stpv2}
\end{algorithm}
Let $s_{*}$ denote
the critical sparsity of the corresponding STP algorithm in the Gaussian signal case. In the above modified version STPv2, we add a new parameter $\gamma$ which can be selected  as $\frac{s_{*}}{m}$. By this modification, the
cardinality of $\tS^{n}$ will always be equal or less than $\lceil 2 \gamma m\rceil$, so the condition number of $\bPhi_{\tilde{S}^n}^{\! *}\bPhi_{\tilde{S}^n}$ will always be bounded by some reasonable value.
\begin{figure}[th]
\centering
\subfigure[Exact reconstruction rate under $150\times300$ Gaussian measurement matrix]{
\label{fg:stpv2-test:a}
\centering
\includegraphics[scale=0.5]{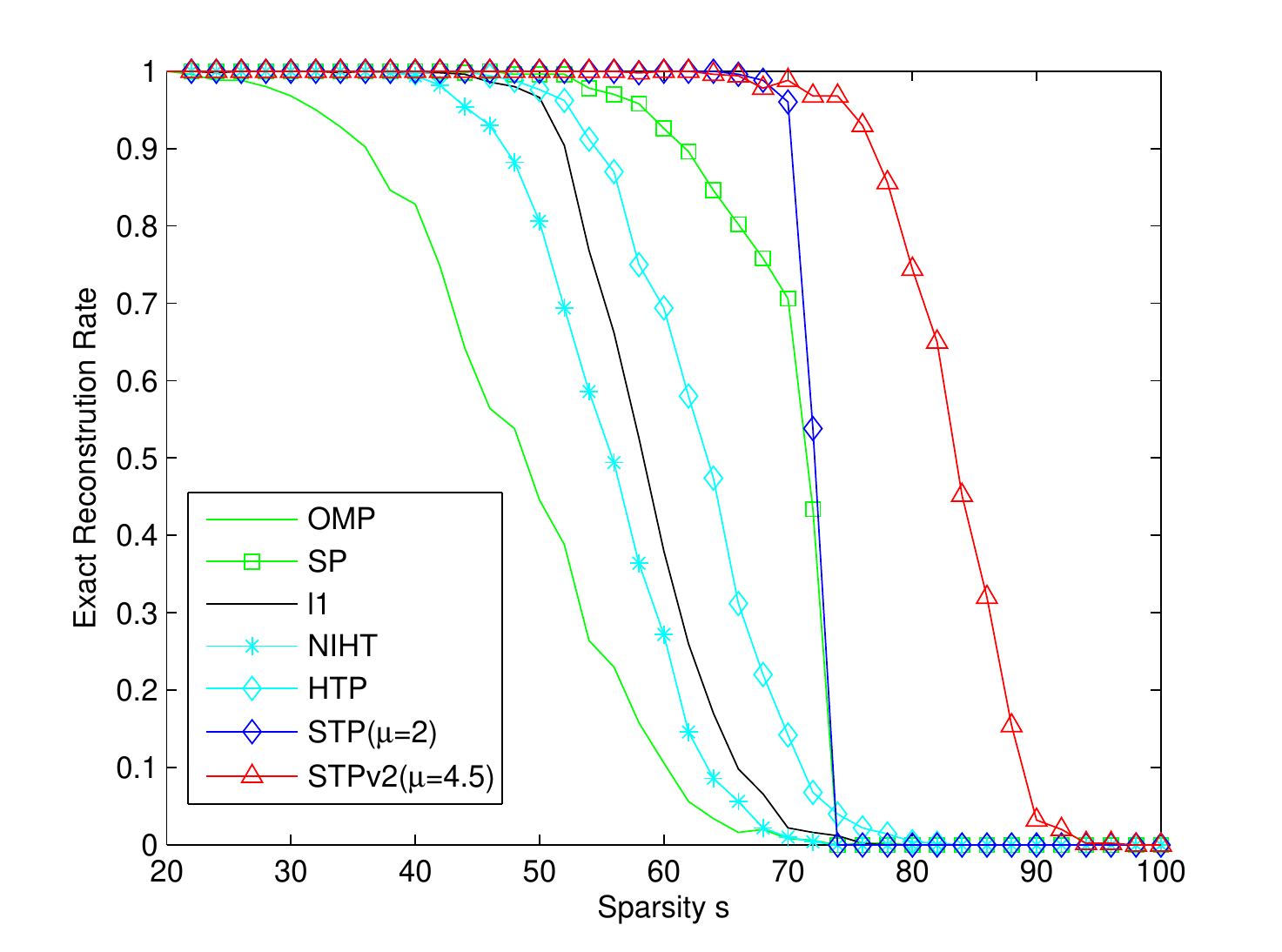}
}
\subfigure[Exact reconstruction rate under $210\times300$ Gaussian measurement matrix]{
\label{fg:stpv2-test:b}
\centering
\includegraphics[scale=0.47]{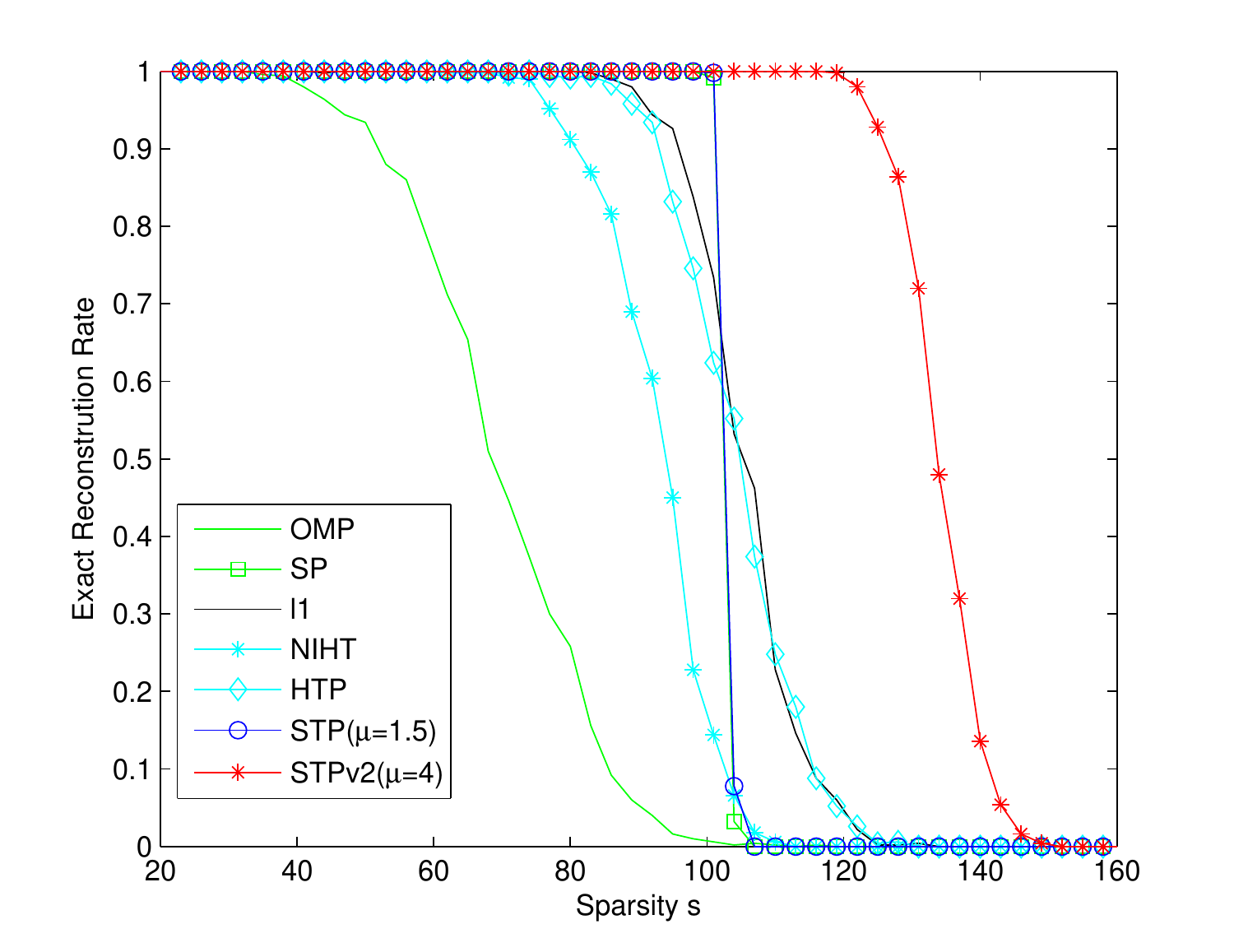}
}
\caption{The performance of STPv2 in the Gaussian signal case.}\label{fg:stpv2-test}
\end{figure}

The two subfigures in Fig. \ref{fg:stpv2-test} show the effect of this modification.
In the two subfigures, we use the Gaussian measurement matrix of size  $150\times300$, $210\times300$, perform 500 independent trails and show the curve of STP  and STPv2  with the corresponding $\mu_{*}$ respectively.
In Fig. \ref{fg:stpv2-test:a}, the critical sparsity of STPv2 starts to exceed the thresholding $\lceil m/2\rceil$ which is the bottleneck of all iterative
greedy algorithms that needs to keep a support set with size being equal or greater than $2s$ for solving a least-squares problems, such as SP, CoSaMP, STP.
In Fig. \ref{fg:stpv2-test:b}, the critical sparsity of STPv2 exceeds $\lceil m/2\rceil$ a lot.

The modification has little impact on the reconstruction capability in the CARS signal case since unless $\tau$ is too large (maybe meaningless in practice), the reconstruction capability of the STP's greedy strategy will
attain its limit before the condition number of $\bPhi_{\tilde{S}^n}^{\! *}\bPhi_{\tilde{S}^n}$ influence the improvement of critical sparsity; in addition, the critical sparsity in the CARS signal case is smaller than that in the Gaussian signal
case, so the reconstruction capability in the CARS signal case will not be impacted naturally if  we select $\gamma$ as $\frac{s_{*}}{m}$.

\subsection{\label{sub:iht-like}Improving performance of other greedy algorithms by IHT-like identification}

We can modify  CoSaMP, HTP, SAMP and FBP by adding the IHT-like identification step in a
suitable step and denote the resulting algorithms as CoSaMPv2, HTPv2, SAMPv2 and FBPv2 respectively. As we say in Section \ref{sec:Introduction}, CoSaMP and HTP are representatives of iterative greedy algorithms, while SAMP and FBP can be treated as
two different kinds of sparsity adaptive versions of SP, which are suitable to the situation that the sparsity $s$ is unknown.
\footnote{These algorithms do not need
the sparsity $s$ as their parameters which is useful, but we should notice that if we need to reconstruct exact $s$-sparsity signals in the noiseless case with probability $1$, the parameter $s$ can be simply set as the critical sparsity which can be tested a priori.
.}
The main steps of the modified versions CoSaMPv2, HTPv2, SAMPv2 and FBPv2 are summarized in Alg. \ref{alg:cosampv2}, \ref{alg:htpv2}, \ref{alg:sampv2} and \ref{alg:fbpv2}
respectively.

\begin{algorithm}[tH]
Input: $\y,\bPhi,s,\alpha,\mu$.\\
Initialization: $S^{0}=\emptyset,\x^0=\0$.\\
Iteration: At the $n$-th iteration, go through the following steps.
\begin{enumerate}
\item $\Delta S =$ \{$\alpha s$ indices corresponding to the  $\alpha s$ largest magnitude entries in the vector $\bPhi^{\! *}\, (\y- \bPhi \x^{n-1})$\}.
\item $\tS^{n}=S^{n-1} \bigcup \Delta S$.
\item $\tx^{n}=\text{arg}\min_{\z \in\bbR^N}\{ \Vert\y-\bPhi\z\Vert_2,\;
    \text{supp}(\z)\subseteq \tS^{n}\}$.

\item $U^{n}=$\{$s$ indices corresponding to the $s$ largest magnitude elements of $\mathbf{\tilde{x}}^{n}$\}.
\item $\mathbf{u}^{n}=$ \{the vector from ${\mathbf{\tilde{x}}}^{n}$ that keeps the entries of ${\mathbf{\tilde{x}}}^{n}$ in $U^{n}$ and set all other ones to zero.\}
\item $\x^{n}=$\{the vector  that keeps the $s$ largest magnitude entries of $\u^{n}+\mu\bPhi^{\! *}(\y-\bPhi\u^{n})$ and set all other ones to zero.\}

\end{enumerate}
until the stopping criteria is met. \\
Output: $\mathbf{x}^{n}$, \text{supp}($\mathbf{x}^{n}$).
\caption{Compressive sampling matching pursuit version  2 (CoSaMPv2)}\label{alg:cosampv2}
\end{algorithm}

\begin{algorithm}[tH]
Input: $\y,\bPhi,s,\alpha,\mu^{\prime},\mu$.\\
Initialization: $S^{0}=\emptyset,\x^0=\0$.\\
Iteration: At the $n$-th iteration, go through the following steps.
\begin{enumerate}
\item $\tS^{n} =$ \{$(\alpha+1) s$ indices corresponding to the  $(\alpha+1) s$ largest magnitude entries in the vector $\x^{n-1}+\mu^{\prime}\bPhi^{\! *}(\y-\bPhi\x^{n-1})$\}.
\item $\tx^{n}=\text{arg}\min_{\z \in\bbR^N}\{ \Vert\y-\bPhi\z\Vert_2,\;
    \text{supp}(\z)\subseteq \tS^{n}\}$.

\item $U^{n}=$\{$s$ indices corresponding to the $s$ largest magnitude elements of $\mathbf{\tilde{x}}^{n}$\}.
\item $\mathbf{u}^{n}=$ \{the vector from ${\mathbf{\tilde{x}}}^{n}$ that keeps the entries of ${\mathbf{\tilde{x}}}^{n}$ in $U^{n}$ and set all other ones to zero.\}
\item $S^{n}=$\{$s$ indices correspoding to the $s$ largest magnitude entries of $\u^{n}+\mu\bPhi^{\! *}(\y-\bPhi\u^{n})$\}.
\item $\x^{n}=\text{arg}\min_{\z \in\bbR^N}\{ \Vert\y-\bPhi\z\Vert_2,\;
    \text{supp}(\z)\subseteq S^{n}\}$.

\end{enumerate}
until the stopping criteria is met. \\
Output: $\mathbf{x}^{n}$, \text{supp}($\mathbf{x}^{n}$).
\caption{Hard Thresholding Pursuit version 2 (HTPv2)}\label{alg:htpv2}
\end{algorithm}

\begin{algorithm}[tH]
Input: $\y,\bPhi,$ $\nu_0$, $\mu$.\\
Initialization: $S^{0}=\emptyset,\x^0=\0,\nu=\nu_0$.\\
Iteration: At the $n$-th iteration, go through the following steps.
\begin{enumerate}
\item $\Delta S =$ \{$\nu$ indices corresponding to the  $\nu$ largest magnitude entries in the vector $\bPhi^{\! *}\, (\y- \bPhi \x^{n-1})$\}.
\item $\tS^{n}=S^{n-1} \bigcup \Delta S$.
\item $\tx^{n}=\text{arg}\min_{\z \in\bbR^N}\{ \Vert\y-\bPhi\z\Vert_2,\;
    \text{supp}(\z)\subseteq \tS^{n}\}$.

\item $U^{n}=$\{$s$ indices corresponding to the $s$ largest magnitude elements of $\mathbf{\tilde{x}}^{n}$\}.
\item $\mathbf{u}^{n}=$ \{the vector from ${\mathbf{\tilde{x}}}^{n}$ that keeps the entries of ${\mathbf{\tilde{x}}}^{n}$ in $U^{n}$ and set all other ones to zero.\}
\item $V=$\{$s$ indices correspoding to the $s$ largest magnitude entries of $\u^{n}+\mu\bPhi^{\! *}(\y-\bPhi\u^{n})$\}.
\item $\v=\text{arg}\min_{\z \in\bbR^N}\{ \Vert\y-\bPhi\z\Vert_2,\;
    \text{supp}(\z)\subseteq V\}$. \\
\textbf{if} the stopping criteria true then \\
quit the iteration; \\
\textbf{elseif} $\|\y-\bPhi\v\|_2\ge\|\y-\bPhi\x^{n-1}\|_2$ then \\
$\nu=\nu+\nu_0$; \\
\textbf{else} \\
$S^n=V$; \\
$\x^{n}=\v$; \\
\textbf{end if}
\end{enumerate}
Output: $\mathbf{x}^{n}$, \text{supp}($\mathbf{x}^{n}$).
\caption{Sparsity Adaptive Matching Pursuit version 2 (SAMPv2)}\label{alg:sampv2}
\end{algorithm}

\begin{algorithm}[tH]
Input: $\y,\bPhi,\mu, \nu,\chi$.\\
Initialization: $S^{0}=\emptyset,\x^0=\0$.\\
Iteration: At the $n$-th iteration, go through the following steps.
\begin{enumerate}
\item $\Delta S =$ \{$\nu$ indices corresponding to the  $\nu$ largest magnitude entries in the vector $\bPhi^{\! *}\, (\y- \bPhi \x^{n-1})$\}.
\item $\tS^{n}=S^{n-1} \bigcup \Delta S$.
\item $\tx^{n}=\text{arg}\min_{\z \in\bbR^N}\{ \Vert\y-\bPhi\z\Vert_2,\;
    \text{supp}(\z)\subseteq \tS^{n}\}$.

\item $U^{n}=$\{$|\tilde{S}^n|-\chi $ indices corresponding to the $|\tilde{S}^n|-\chi $ largest magnitude elements of $\mathbf{\tilde{x}}^{n}$\}.
\item $\mathbf{u}^{n}=$ \{the vector from ${\mathbf{\tilde{x}}}^{n}$ that keeps the entries of ${\mathbf{\tilde{x}}}^{n}$ in $U^{n}$ and set all other ones to zero.\}
\item $S^{n}=$\{$|U^{n}|$ indices correspoding to the $|U^{n}|$ largest magnitude entries of $\u^{n}+\mu\bPhi^{\! *}(\y-\bPhi\u^{n})$\}.
\item $\x^{n}=\text{arg}\min_{\z \in\bbR^N}\{ \Vert\y-\bPhi\z\Vert_2,\;
    \text{supp}(\z)\subseteq S^{n}\}$.

\end{enumerate}
until the stopping criteria is met. \\
Output: $\mathbf{x}^{n}$, \text{supp}($\mathbf{x}^{n}$).
\caption{Forward Backward  Pursuit version 2 (FBPv2)} \label{alg:fbpv2}
\end{algorithm}

\begin{figure}[th]
\centering
\subfigure[Exact reconstruction rate of CoSaMPv2 and HTPv2  in the Gaussian signal case]{
\label{fg:cosamp-htp-test:a}
\centering
\includegraphics[scale=0.5]{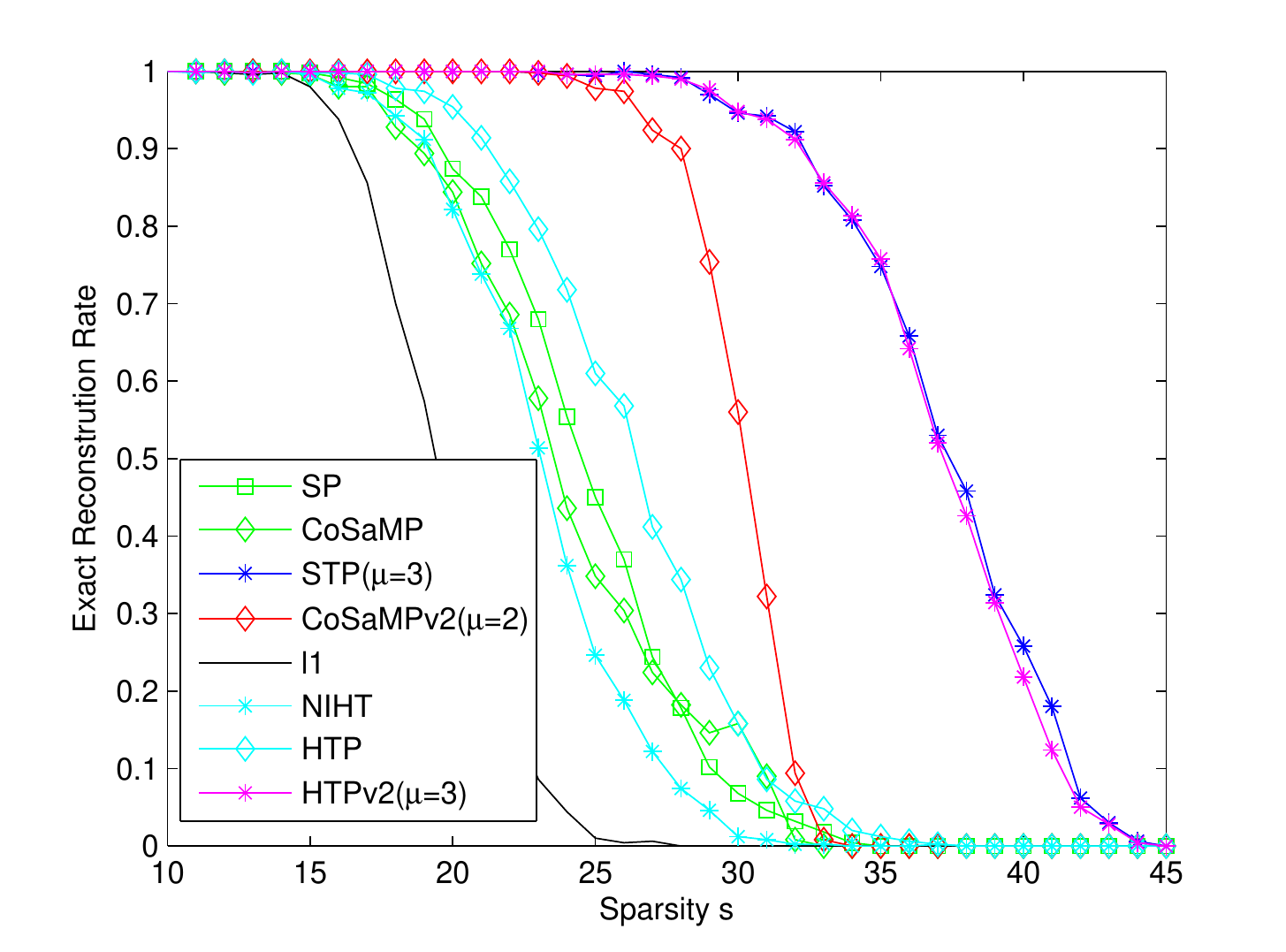}
}
\subfigure[Exact reconstruction rate of CoSaMPv2 and HTPv2 in the CARS signal case]{
\label{fg:cosamp-htp-test:b}
\centering
\includegraphics[scale=0.5]{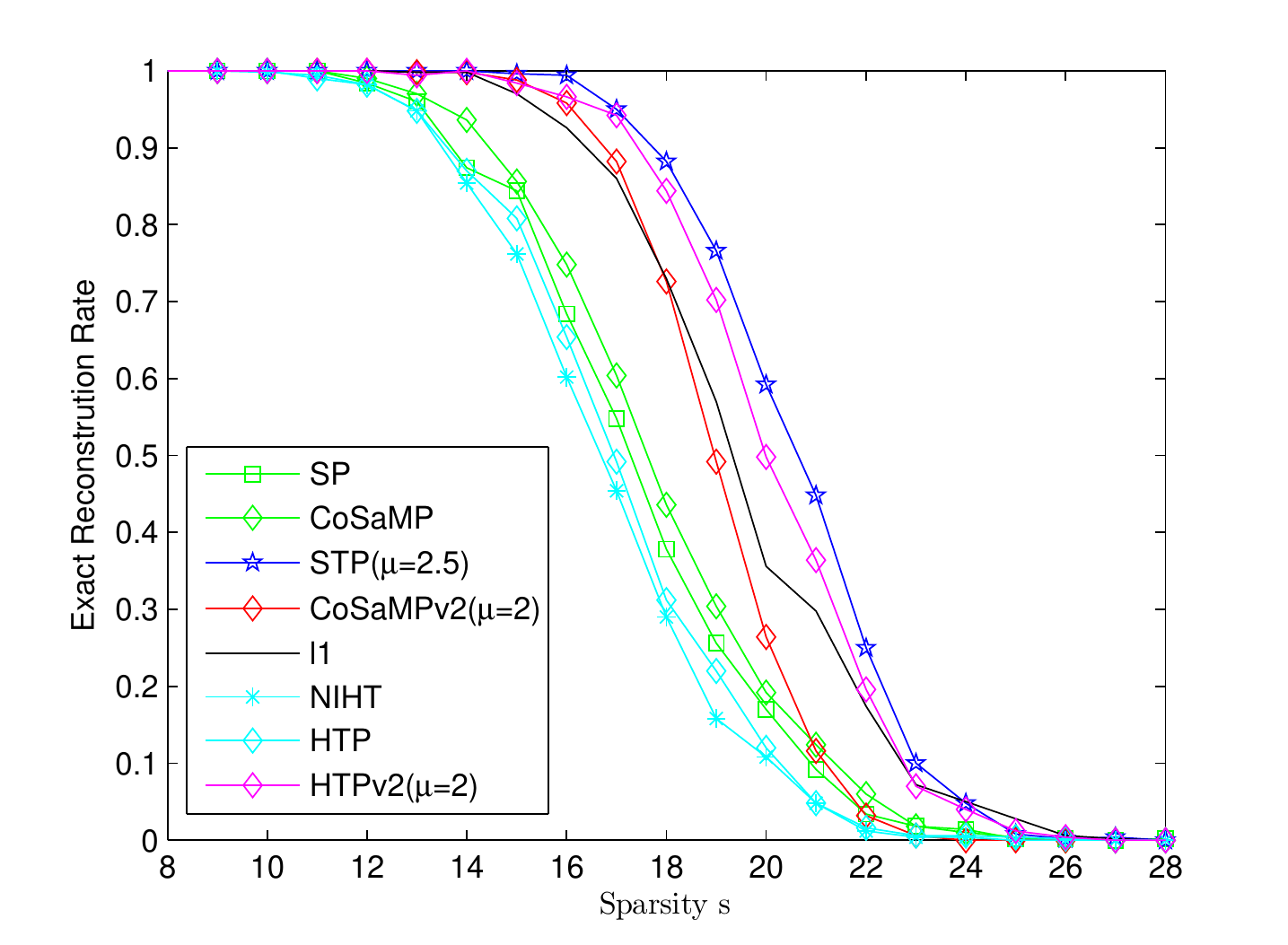}
}
\caption{The performance of CoSaMPv2 and HTPv2 under $100\times1000$ Gaussian measurement matrix with undersampling ratio $\tau=0.1$.}\label{fg:cosamp-htp-test}
\end{figure}
\begin{figure}[th]
\centering
\subfigure[Exact reconstruction rate of SAMPv2 and FBPv2 in the Gaussian signal case]{
\label{fg:samp-fbp-test:a}
\centering
\includegraphics[scale=0.5]{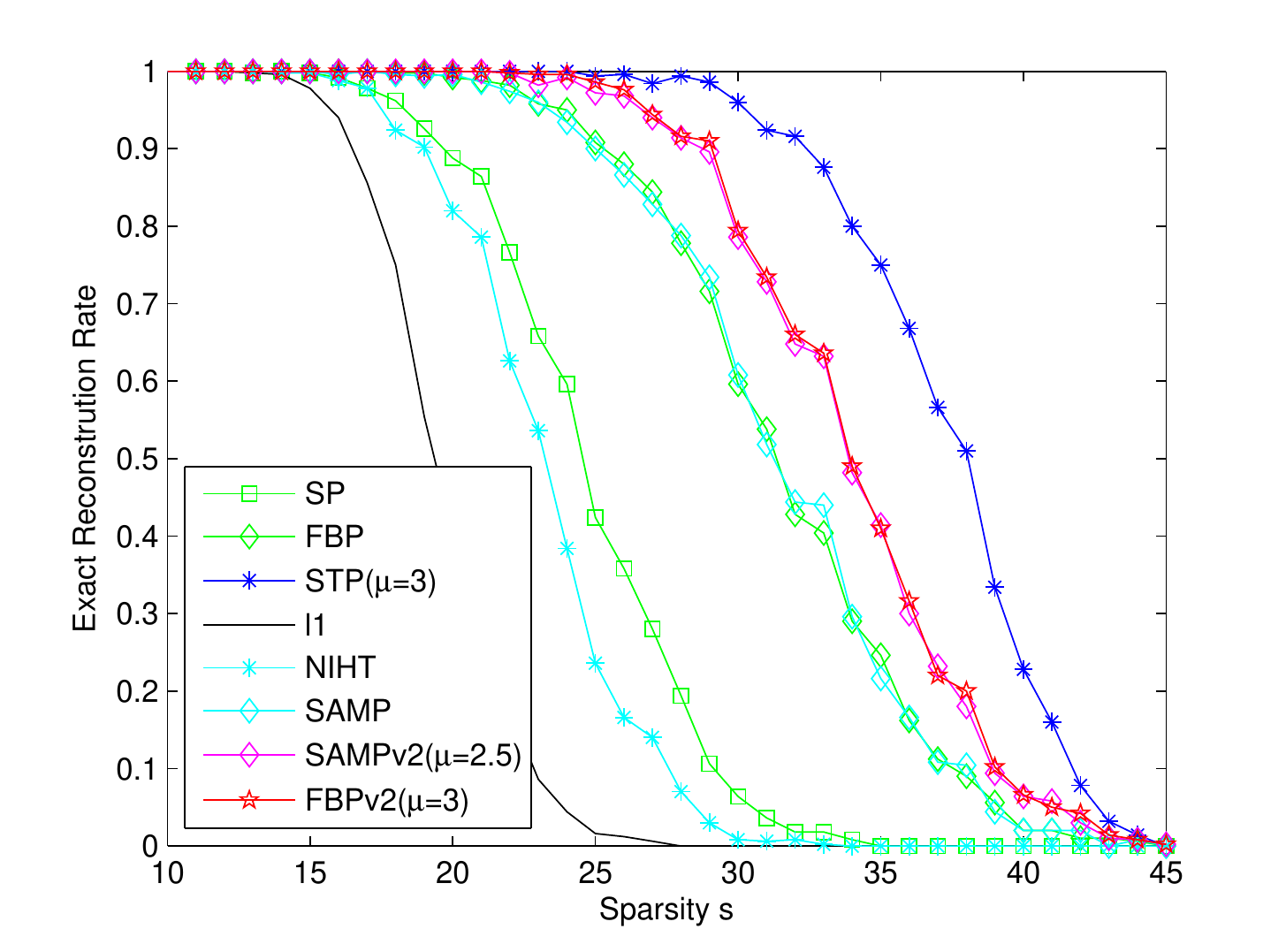}
}
\subfigure[Exact reconstruction rate of SAMPv2 and FBPv2 in the CARS signal case]{
\label{fg:samp-fbp-test:b}
\centering
\includegraphics[scale=0.5]{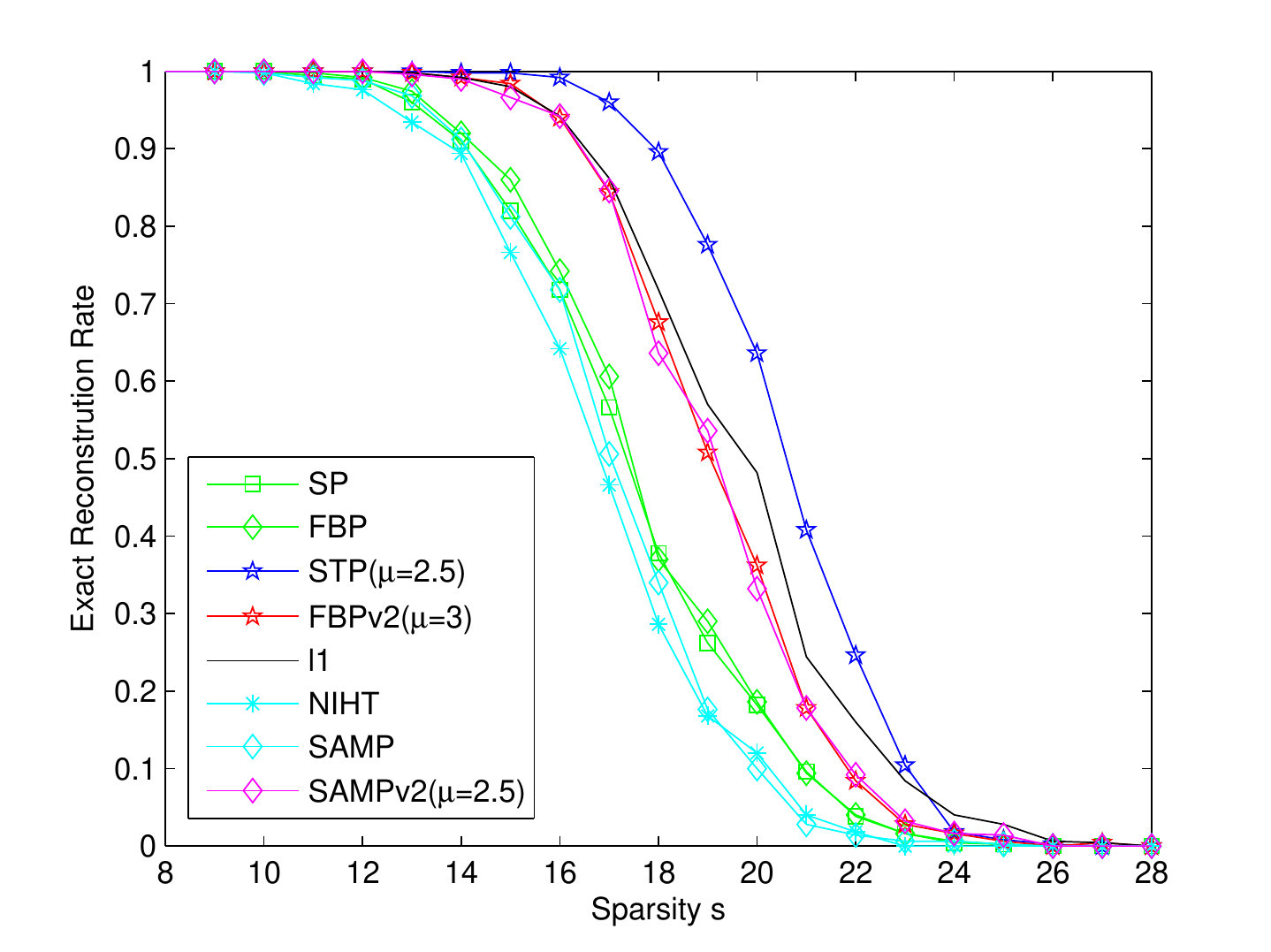}
}
\caption{The performance of SAMPv2 and FBPv2 under $100\times1000$ Gaussian measurement matrix with undersampling ratio $\tau=0.1$.}\label{fg:samp-fbp-test}

\end{figure}

The detailed descriptions of the original algorithms can be seen in \cite{karahanoglu2013compressed,needell2009cosamp,foucart2011hard,do2008sparsity}. In the following paragraph,
we only give the descriptions of the changes in CoSaMPv2, HTPv2, SAMPv2 and FBPv2 compared with the original algorithms as well as the settings in our simulations.

In each trail, we use a $100\times1000$ Gaussian measurement matrix and perform 500 independent trials respectively.
 In CoSaMPv2, we add the IHT-like identification step in step 6 and  set $\alpha=2$ in our simulations to imitate the initial version of CoSaMP in \cite{needell2009cosamp}
(the more general description of CoSaMP can be seen in \cite{tropp2010computational});
In HTPv2, steps 1 and 5 are IHT-like identification steps and the parameters $\mu$ and $\mu^{\prime}$ can be adjusted in practice. Let $|\v|_{\min}, |\v|_{\max}$ denote the smallest magnitude entries and the largest
entries in arbitrary vector $\v$. Due to $(\bPhi^{\! * }(\y-\bPhi\x^{n-1}))_{S^{n-1}}=\0$, when $\mu^{\prime}<\frac{|\x^{n-1}|_{\min}}{|\bPhi^{\! *}(\y-\bPhi\x^{n-1})|_{\max}}$,
in step 1 of HTPv2, $\tS^{n}$ will contain supp($\x^{n-1}$) at first, then HTPv2 degrades to
STP. In addition, if we set $\alpha=0,\mu^{\prime}=\mu=1$, HTPv2 degrades to HTP. In order to address the general case, we set $\alpha=1,\mu^{\prime}=1$ in our simulations.
In SAMPv2, we add the IHT-like identification step in steps 5 and 6 and set $\nu_0=2$ in our simulations; in FBPv2, we add the IHT-like identification steps 5 and 6
and set $\nu=20, \chi=18$.

The  four subfigures in Fig.  \ref{fg:cosamp-htp-test} and \ref{fg:samp-fbp-test} show their empirical performance. We show the performance of CoSaMPv2 and HTP v2 in Fig. \ref{fg:cosamp-htp-test:a},
\ref{fg:cosamp-htp-test:b} and then the performance of SAMPv2 and FBPv2 in Fig. \ref{fg:samp-fbp-test:a},
\ref{fg:samp-fbp-test:b}. In each figure, we display the curve of each algorithm with $\mu_*$ if it needs the parameter $\mu$.
We show that IHT-like identification step is a universal way to improve the empirical performance, but
so far, in all of these improvements, the improvement to SP, i.e., the proposed STP algorithm has the best empirical performance.

\section{\label{sec:conclusion}Conclusion}
In this paper, we proposed a new reconstruction algorithm for CS, termed subspace thresholding pursuit (STP). STP has a strong provable theoretical guarantee and good empirical performance.
It digs out the potential of iterative greedy algorithms further and displays an outstanding worst-case empirical performance which is better than the well-known $\ell_1$ minimization if the undersampling ratio is not very large. In addition, we proposed a method to improve the critical sparisity further if the undersampling ratio is large and  showed the universal significance of the idea in STP. Future works may focus on solving the inconsistency of the theoretical guarantee and empirical performance of the paper, e.g., STP with some parameter $\mu>1$ may have better empirical performance but worse theoretical guarantee than the one with $\mu=1$ in the current version.

\appendix
%

Considering the similarity to SP, our theoretical analysis for STP mainly follows the framework the authors developed in \cite{song2013improved}.
Before our derivations, we firstly introduce some lemmas which are mainly referenced or developed in \cite{song2013improved}.
\subsection{Some technical lemmas}
 The following two lemmas are used in the derivations of RIC related results.
\begin{lemma}[Consequences of the RIP]$ $
\label{lem:rip}
\begin{enumerate}
\item (Monotonicity \cite{candes2005decoding}) For any two positive integers $s\le s^{\prime}$, \quad $\delta_s\le\delta_{s^{\prime}}.$
\item For two vectors $\mathbf{u}, \mathbf{v}\in\mathbb{R}^{N}$, if $|$supp($\mathbf{u}$)$\cup$supp($\mathbf{v}$)$|$$\le t$, then
\begin{eqnarray}
        |\langle\mathbf{u}, (\mathbf{I}-\mu\mathbf{\Phi}^{\! *}\mathbf{\Phi})\mathbf{v}\rangle|\le(|\mu-1|+\mu\delta_t)\Vert\mathbf{u}\Vert_2 \Vert\mathbf{v}\Vert_2;\label{rip11}
\end{eqnarray}
      moreover, if $U\subseteq\{1,\dots,N\}$ and $|U \cup \text{supp}(\mathbf{v})$$|$$\le t$, then
\begin{eqnarray}
            \Vert((\mathbf{I}-\mu\mathbf{\Phi}^{\! *}\mathbf{\Phi})\mathbf{v})_U\Vert_2\le(|\mu-1|+\mu\delta_t)\Vert\mathbf{v}\Vert_2.\label{rip12}
\end{eqnarray}
We omit  the proofs of \eqref{rip11} and \eqref{rip12} here for their  similarity to the proofs of \cite[Lemma 1]{song2013improved}.
\end{enumerate}
\end{lemma}
\vspace{0.1in}
\begin{lemma}[Noise perturbation in partial support \cite{foucart2011hard}]
\label{lem:noise}
For the general CS model $
\mathbf{y}=\mathbf{\Phi}\mathbf{x}_S+\mathbf{e}^{\prime}$ in (\ref{eq:general_model}), letting $U\subseteq\{1,\ldots,N\}$ and $|U|\le u$, we have
\begin{eqnarray}
\label{rip13}
\Vert(\mathbf{\Phi}^{\! *}\mathbf{e}^{\prime})_{U}\Vert_2\le\sqrt{1+\delta_{u}}\Vert\mathbf{e}^{\prime}\Vert_2.
\end{eqnarray}
\end{lemma}
The next lemma introduces a simple inequality introduced in \cite{song2013improved} which is useful in our derivations.
\begin{lemma}[\cite{song2013improved}]
\label{lem:cauchy}
For nonnegative numbers $a,b,c,d,x,y$,
\begin{align}
\label{lem3}
(ax+by)^2+(cx+dy)^2\le (\sqrt{a^2+c^2}x+(b+d)y)^2.
\end{align}
\end{lemma}

Consider the general CS model $
\mathbf{y}=\mathbf{\Phi}\mathbf{x}_S+\mathbf{e}^{\prime}$ in (\ref{eq:general_model}). Let $T\subseteq \{1,2,\ldots,N\}$ and $|T|=t$. Let $\mathbf{z}_{p}$ be the solution of the least squares problem $\mbox{arg}\min_{\mathbf{z}\in\mathbb{R}^{N}}\{\Vert\mathbf{y}-
\mathbf{\Phi}\mathbf{z}\Vert_2, \;\text{supp}(\mathbf{z}$)$\subseteq T\}$. The least squares problem has the following orthogonal properties introduced in \cite{song2013improved}.

\begin{lemma}[Consequences for orthogonality by the RIP \cite{song2013improved}]
\label{lem:orthogonality-rip}
If $\delta_{s+t}<1$,
\begin{equation}
\Vert(\mathbf{x}_S-\mathbf{z}_{p})_T\Vert_2\le\delta_{s+t}\Vert\mathbf{x}_S-\mathbf{z}_{p}\Vert_2+\sqrt{1+\delta_{t}}\Vert\mathbf{e}^{\prime}\Vert_2\label{eq:orthogonality-rip1}
\end{equation}
and
\begin{equation}
\Vert\mathbf{x}_S-\mathbf{z}_{p}\Vert_2\le\sqrt{\dfrac{1}{1-\delta_{s+t}^2}}
\Vert(\mathbf{x}_S)_{\overline{T}}\Vert_2+\dfrac{\sqrt{1+\delta_{t}}}
{1-\delta_{s+t}}\Vert\mathbf{e}^{\prime}\Vert_2.\label{eq:orthogonality-rip2}
\end{equation}
Moveover, if $t>s$, define $T_{\nabla}$:=\{The indices of the $t-s$ smallest magnitude entries of $\mathbf{z}_{p}$ in $T$\}, we have
\begin{eqnarray}
\Vert(\mathbf{x}_S)_{T_{\nabla}}\Vert_2
\le\sqrt{2}\delta_{s+t}\Vert\mathbf{x}_S
-\mathbf{z}_{p}\Vert_2+\sqrt{2(1+\delta_{t})}\Vert\mathbf{e}^{\prime}\Vert_2.
\label{eq:orthogonality-rip3}
\end{eqnarray}
\end{lemma}

Throughout the paper, we use the notation $(i)$ stacked over an inequality sign to indicate
that the inequality follows from the expression $(i)$ in the paper.

Before our analysis, we emphasize again that unless stated, we set $\alpha=1$ in this paper.

In Alg. \ref{alg:stp}, the first two steps and the last step are identical with the corresponding steps of SP in Alg. \ref{alg:sp}, so the property of the identification step for SP is also suitable for STP. Similar to \cite[Lemma 6]{song2013improved}, we have a lemma for STP as follows.

\begin{lemma}In the steps 1 and 2 of STP, we have
\label{lem:identification-sp}
\[\Vert(\mathbf{x}_S)_{\overline{\tilde{S}^{n}}}\Vert_2\le\sqrt{2}\delta_{3s}\Vert\mathbf{x}_S- \mathbf{x}^{n-1}\Vert_2+\sqrt{2(1+\delta_{2s})}\Vert\mathbf{e}^{\prime}\Vert_2.
\]
\end{lemma}

In Alg. \ref{alg:stp}, the step 6, i.e., the IHT-like identification step also has similar property with steps 1 and 2, which has been developed in Foucart \cite{foucart2011hard}  when $\mu=1$. In our discussion, we generalize the result in \cite{foucart2011hard} to the more general case with $\mu\ge0$.

\begin{lemma}[IHT-like Identification] In the step 6 of STP, we have
\label{lem:iht-identification}
\begin{eqnarray*}
&&\Vert(\mathbf{x}_S)_{\overline{S^{n}}}\Vert_2\nonumber\\
&&\le\sqrt{2}(|\mu-1|+\mu\delta_{3s})\Vert\mathbf{x}_S- \u^{n}\Vert_2+\sqrt{2(1+\delta_{2s})}\mu\Vert\mathbf{e}^{\prime}\Vert_2.
\end{eqnarray*}
\end{lemma}
\begin{proof}
  In the step 6 of the $n$-th iteration, $S^n$ is the set of the $s$ indices corresponding to the  $s$ largest magnitude entries in $ \u^n+\mu\bPhi^{\! *} (\y-\bPhi\u^n)$. Thus,
\begin{eqnarray}
&&\|(\u^n+\mu\bPhi^{\! *} (\y-\bPhi\u^n))_S\|_2\nonumber\\
&&\quad\quad\le\|(\u^n+\mu\bPhi^{\! *} (\y-\bPhi\u^n))_{S^n}\|_2.
\end{eqnarray}
Removing the common coordinates in $S \cap S^n$ and noticing that $\mathbf{y}=\mathbf{\Phi}\mathbf{x}_S+\mathbf{e}^{\prime}$, we have
\begin{eqnarray}
&&\|(\u^n+\mu\bPhi^{\! *}(\y-\bPhi\u^n))_{S\backslash S^n}\|_2\nonumber\\
&&\le\|(\u^n+\mu\bPhi^{\! *} (\y-\bPhi\u^n))_{S^n\backslash S}\|_2. \label{eq:stp-61}
\end{eqnarray}
For the right-hand side of (\ref{eq:stp-61}), noticing that $(\mathbf{x}_S)_{S^n\backslash S}=\mathbf{0}$, we have
\begin{eqnarray}
&&\|(\u^n+\mu\bPhi^{\! *} (\y-\bPhi\u^n))_{S^n\backslash S}\|_2\nonumber\\
&=&\|(\u^n+\mu\bPhi^{\! *}\bPhi (\x_{S}-\u^n)+\mu\bPhi^{\! *} \e^{\prime})_{S^n\backslash S}\|_2 \nonumber\\
&\le&\!\!\!\!\!\!\|((\mu\bPhi^{\! *}\bPhi-\I)(\x_{S}-\u^n))_{S^n\backslash S}\|_2+\|(\mu\bPhi^{\! *} \e^{\prime})_{S^n\backslash S}\|_2. \label{eq:stp-62}
\end{eqnarray}
For the left-hand side of (\ref{eq:stp-61}), noticing that $(\mathbf{x}_S)_{S\backslash S^{n}}=(\mathbf{x}_S)_{\overline{ S^n}}$, we have
\begin{eqnarray}
&&\|(\u^n+\mu\bPhi^{\! *} (\y-\bPhi\u^n))_{S\backslash S^n}\|_2\nonumber\\
&=&\|(\u^n+\mu\bPhi^{\! *}\bPhi (\x_{S}-\u^n)+\mu\bPhi^{\! *} \e^{\prime}-\x_{S}+\x_{S})_{S\backslash S^n}\|_2\nonumber\\
&\ge&\|(\x_{S})_{\overline{S^n}}\|_2-\|((\mu\bPhi^{\! *}\bPhi-\I)(\x_{S}-\u^n))_{S\backslash S^n}\|_2\nonumber\\
&&-\|(\mu\bPhi^{\! *} \e^{\prime})_{S\backslash S^n}\|_2.  \label{eq:stp-63}
\end{eqnarray}
Combining (\ref{eq:stp-61}), (\ref{eq:stp-62}) and (\ref{eq:stp-63}), we have
\begin{eqnarray}
&&\|(\x_S)_{\overline{S^n}}\|_2\nonumber\\
&\le&\|((\mu\bPhi^{\! *}\bPhi-\I)(\x_{S}-\u^n))_{S^n\backslash S}\|_2+\|(\mu\bPhi^{\! *} \e^{\prime})_{S^n\backslash S}\|_2 \nonumber\\
&&+\|((\mu\bPhi^{\! *}\bPhi-\I)(\x_{S}-\u^n))_{S\backslash S^n}\|_2+\|(\mu\bPhi^{\! *} \e^{\prime})_{S\backslash S^n}\|_2 \nonumber\\
&\le&\sqrt{2}\|((\mu\bPhi^{\! *}\bPhi-\I)(\x_{S}-\u^n))_{(S^n\backslash S)\cup(S\backslash S^n)}\|_2\nonumber\\
&&+\sqrt{2}\|(\mu\bPhi^{\! *}  \e^{\prime})_{(S^n\backslash S)\cup(S\backslash S^n)}\|_2 \label{eq:stp-64}\\
&\le&\sqrt{2}(|\mu-1|+\mu\delta_{3s})\|\x_{S}-\u^n\|_2\nonumber\\
&&+\sqrt{2(1+\delta_{2s})}\mu\|\e^{\prime}\|_2.
\end{eqnarray}
where the inequality $(\ref{eq:stp-64})$ is from the Cauchy-Schwartz inequality.
\end{proof}

\subsection{\label{proof1}Proof of Theorem \ref{thm:main}}
Steps 1 and 2
 are the OMP-like identification steps.
By Lemma \ref{lem:identification-sp}, in the $n$-th iteration, we have

\begin{equation}
\Vert(\mathbf{x}_S)_{\overline{\tilde{S}^{n}}}\Vert_2\le\sqrt{2}\delta_{3s}\Vert\mathbf{x}_S- \mathbf{x}^{n-1}\Vert_2+\sqrt{2(1+\delta_{2s})}\Vert\mathbf{e}^{\prime}\Vert_2.\label{eq:sp_1}
\end{equation}

Step 3 of the $n$-th iteration is a procedure of solving a least squares problem. Letting $T=\tilde{S}^{n}$ and $\mathbf{z}_p=\tilde{\mathbf{x}}^{n}$, $t=2s$, by (\ref{eq:orthogonality-rip2}) of Lemma \ref{lem:orthogonality-rip}, we have
\begin{align}
&\Vert\mathbf{x}_S-\tx^{n}\Vert_2\le\sqrt{\dfrac{1}{1-\delta_{3s}^2}}\Vert(\mathbf{x}_S)_{\overline{\tilde{S}^{n}}}\Vert_2+\dfrac{\sqrt{1+\delta_{2s}}}{1-\delta_{3s}}\Vert\mathbf{e}^{\prime}\Vert_2.\label{eq:sp_2}
\end{align}
Then combining (\ref{eq:sp_1}) and (\ref{eq:sp_2}) and magnifying $\delta_{2s}$ to $\delta_{3s}$ by Lemma \ref{lem:rip}, we have
\begin{eqnarray}
\Vert\mathbf{x}_S-\tx^{n}\Vert_2
&\le&\sqrt{\dfrac{2\delta_{3s}^2}{1-\delta_{3s}^2}}\Vert\mathbf{x}_S-\mathbf{x}^{n-1}\Vert_2\nonumber\\
&&+\dfrac{\sqrt{2(1-\delta_{3s})}+\sqrt{1+\delta_{3s}}}{1-\delta_{3s}}\Vert\mathbf{e}^{\prime}\Vert_2.\label{eq:sp_3}
\end{eqnarray}

In step 4 of the $n$-th iteration, define $S_{\nabla}:=\tilde{S}^{n}\backslash U^{n}$, where $S_{\nabla}$ contains the indices of the $s$ smallest entries in $\mathbf{\tilde{x}}^{n}$. Letting $T=\tilde{S}^{n}$ and $\mathbf{z}_p=\tx^{n}$, $t=2s$, $T_{\nabla}=S_{\nabla}$, by (\ref{eq:orthogonality-rip3}) of Lemma \ref{lem:orthogonality-rip}, we have
\begin{align}
\Vert(\mathbf{x}_S)_{S_{\nabla}}\Vert_2\le \sqrt{2}\delta_{3s}\Vert\mathbf{x}_S-\tx^{n}\Vert_2+\sqrt{2(1+\delta_{2s})}\Vert\mathbf{e}^{\prime}\Vert_2.\label{eq:sp_4}
\end{align}
Let $\tau_1=\dfrac{\sqrt{2(1-\delta_{3s})}+\sqrt{1+\delta_{3s}}}{1-\delta_{3s}}$ and $\tau_2=\sqrt{1+\delta_{3s}}$. Dividing $\overline{U^{n}}$ into two disjoint parts: $S_{\nabla}$ and $\overline{\tilde{S}^{n}}$, we have
\begin{eqnarray*}
&&\Vert(\mathbf{x}_S)_{\overline{U^{n}}}\Vert_2^2
=\Vert(\mathbf{x}_S)_{S_{\nabla}}\Vert_2^2+\Vert(\mathbf{x}_S)_{\overline{\tilde{S}^{n}}}\Vert_2^2\\
&\overset{(\ref{eq:sp_4}), (\ref{eq:sp_1})}{\le}&2(\delta_{3s}\Vert\mathbf{x}_{S}-
\tx^{n}\Vert_2
+\tau_{2}\Vert\mathbf{e}^{\prime}\Vert_2)^2\\
&&+2\left(\delta_{3s}\Vert\mathbf{x}_{S}-\mathbf{x}^{n-1}\Vert_2+\tau_{2}\Vert\mathbf{e}^{\prime}\Vert_2\right)^2\\
&\overset{(\ref{eq:sp_3})}{\le}&2\Bigg(\delta_{3s}\sqrt{\dfrac{2\delta_{3s}^2}{1-\delta_{3s}^2}}\Vert\mathbf{x}_S
-\mathbf{x}^{n-1}\Vert_2\\
&&+(\delta_{3s}\tau_1+\tau_2)\Vert\mathbf{e}^{\prime}\Vert_2\Bigg)^2\\
&&+2\left(\delta_{3s}\Vert\mathbf{x}_S-\mathbf{x}^{n-1}\Vert_2+\tau_2\Vert\mathbf{e}^{\prime}\Vert_2\right)^2
\nonumber\\
&\overset{(\ref{lem3})}{\le}&2\Bigg(\sqrt{\dfrac{2\delta_{3s}^4}{1-\delta_{3s}^2}+\delta_{3s}^2}\;\Vert\mathbf{x}_S-\mathbf{x}^{n-1}\Vert_2\\
&&+\left((\delta_{3s}\tau_1+\tau_2)+\tau_2\right)\Vert\mathbf{e}^{\prime}\Vert_2\Bigg)^2\nonumber\\
&=&2\Bigg(\sqrt{\dfrac{\delta_{3s}^2 (1+\delta_{3s}^2)}{1-\delta_{3s}^2}}\Vert\mathbf{x}_S-\mathbf{x}^{n-1}\Vert_2\\
&&+(\delta_{3s}\tau_1+2\tau_2)\Vert\mathbf{e}^{\prime}\Vert_2\Bigg)^2,\nonumber
\end{eqnarray*}
which implies that
\begin{eqnarray}
\Vert(\mathbf{x}_{S})_{\overline{U^{n}}}\Vert_2
&\le&\sqrt{\dfrac{2\delta_{3s}^2 (1+\delta_{3s}^2)}{1-\delta_{3s}^2}}\Vert\mathbf{x}_S-\mathbf{x}^{n-1}\Vert_2\\
&&+\sqrt{2}(\delta_{3s}\tau_1+2\tau_2)\Vert\mathbf{e}^{\prime}\Vert_2.\label{eq:sp_5}
\end{eqnarray}

In step 5 of the $n$-th iteration, since $\mathbf{u}^{n}$ is obtained by keeping the $s$ largest magnitude entries of $\tx^{n}$, we have
\begin{align}
&\quad\;\Vert(\mathbf{x}_S-\mathbf{u}^{n})_{U^{n}}\Vert_2\le\Vert(\mathbf{x}_S-\tx^{n})_{\tilde{S}^{n}}\Vert_2\nonumber\\
&\overset{(\ref{eq:orthogonality-rip1})}{\le}\delta_{3s}\Vert\mathbf{x}_S-\tx^{n}\Vert_2+\sqrt{1+\delta_{3s}}\Vert\mathbf{e}^{\prime}\Vert_2\nonumber\\
&\overset{(\ref{eq:sp_3})}{\le}\sqrt{\dfrac{2\delta_{3s}^4}{1-\delta_{3s}^2}}\Vert\mathbf{x}_S-\mathbf{x}^{n-1}\Vert_2\nonumber\\
&\quad+\left(\delta_{3s}\dfrac{\sqrt{2(1-\delta_{3s})}+\sqrt{1+\delta_{3s}}}{1-\delta_{3s}}+\sqrt{1+\delta_{3s}}\right)\Vert\mathbf{e}^{\prime}\Vert_2\nonumber\\
&=\sqrt{\dfrac{2\delta_{3s}^4}{1-\delta_{3s}^2}}\Vert\mathbf{x}_S-\mathbf{x}^{n-1}\Vert_2+(\delta_{3s}\tau_1+\tau_2)\Vert\mathbf{e}^{\prime}\Vert_2.\label{eq:cosamp_6}
\end{align}
Dividing supp($\mathbf{x}_S-\mathbf{u}^{n}$) into two disjoint parts: $U^{n},\overline{U^{n}}$, and noticing that $(\mathbf{x}_S-\mathbf{u}^{n})_{\overline{U^{n}}}=(\mathbf{x}_S)_{\overline{U^{n}}}$, we have
\begin{eqnarray*}
&&\Vert\mathbf{x}_S-\mathbf{u}^{n}\Vert_2^2=\Vert(\mathbf{x}_S-\mathbf{u}^{n})_{U^{n}}\Vert_2^2+\Vert(\mathbf{x}_S-\mathbf{u}^{n})_{\overline{U^{n}}}\Vert_2^2\nonumber\\
&=&\Vert(\mathbf{x}_S-\mathbf{u}^{n})_{U^{n}}\Vert_2^2+\Vert(\mathbf{x}_S)_{\overline{U^{n}}}\Vert_2^2\nonumber\\
&\overset{(\ref{eq:cosamp_6}),(\ref{eq:sp_5})}{\le}&\Bigg(\sqrt{\dfrac{2\delta_{3s}^4}{1-\delta_{3s}^2}}\Vert\mathbf{x}_S-\mathbf{x}^{n-1}\Vert_2+(\delta_{3s}\tau_1+\tau_2)\Vert\mathbf{e}^{\prime}\Vert_2\Bigg)^2\\
&&+\Bigg(\sqrt{\dfrac{2\delta_{3s}^2 (1+\delta_{3s}^2)}{1-\delta_{3s}^2}}\Vert\mathbf{x}_S-\mathbf{x}^{n-1}\Vert_2\nonumber\\
&&\qquad+\sqrt{2}(\delta_{3s}\tau_1 +2\tau_2)\Vert\mathbf{e}^{\prime}\Vert_2\Bigg)^2\nonumber\\
&\overset{(\ref{lem3})}{\le}&\Bigg(\sqrt{\dfrac{2\delta_{3s}^2(1+2\delta_{3s}^2)}{1-\delta_{3s}^2}}\Vert\mathbf{x}_S-\mathbf{x}^{n-1}\Vert_2\nonumber\\
&&+((\sqrt{2}+1)\delta_{3s}\tau_1+(2\sqrt{2}+1)\tau_2)\Vert\mathbf{e}^{\prime}\Vert_2\Bigg)^2,\label{eq:cosamp_7}
\end{eqnarray*}
or
\begin{eqnarray}
\Vert\mathbf{x}_S-\mathbf{u}^{n}\Vert_2
&\le&\sqrt{\dfrac{2\delta_{3s}^2(1+2\delta_{3s}^2)}{1-\delta_{3s}^2}}\Vert\mathbf{x}_S-\mathbf{x}^{n-1}\Vert_2\nonumber\\
&&\!\!\!+((\sqrt{2}+1)\delta_{3s}\tau_1
+(2\sqrt{2}+1)\tau_2)\Vert\mathbf{e}^{\prime}\Vert_2.\label{eq:stp1}
\end{eqnarray}

Step 6 of STP is the IHT-like identification step.
By Lemma \ref{lem:iht-identification}, in the $n$-th iteration, we have
\begin{eqnarray}
\|(\x_S)_{\overline{S^n}}\|_2&\le&\sqrt{2}(|\mu-1|+\mu\delta_{3s})\|\x_{S}-\u^n\|_2\nonumber\\
&&+\sqrt{2(1+\delta_{2s})}\mu\|\e^{\prime}\|_2.\label{eq:stp2}
\end{eqnarray}

Step 7 of the $n$-th iteration is a procedure of solving a least squares problem. Letting $T=S^{n}$ and $\mathbf{z}_p=\x^{n}$, $t=s$, by (\ref{eq:orthogonality-rip2}) of Lemma \ref{lem:orthogonality-rip}, we have
\begin{equation}
\|\x_S-\x^n\|_2\le\sqrt{\frac{1}{1-\delta_{2s}^2}}\|(\x_S)_{\overline{S^n}}\|_2+\frac{\sqrt{1+\delta_{s}}}{1-\delta_{2s}}\|\e^{\prime}\|_2.\label{eq:stp3}
\end{equation}

Combing \eqref{eq:stp1},\eqref{eq:stp2} and \eqref{eq:stp3},  and magnifying $\delta_{s},\delta_{2s}$ to $\delta_{3s}$ by Lemma \ref{lem:rip}, we have
\begin{eqnarray}
\|\x_S-\x^n\|_2\le\rho\|\x_S-\x^{n-1}\|_2+(1-\rho)\tau\|\e^{\prime}\|_2.
\end{eqnarray}
where $\rho$ and $\tau$ is respectively referred to (\ref{eq:rho}) and (\ref{eq:tau}). Hence, (\ref{eq:main}) follows by recursively using the above inequality when $\rho<1$.

When $\mu<1$, $\rho<1$ is equivalent to $\dfrac{1}{1-\delta_{3s}}-\dfrac{1+\delta_{3s}}{2\delta_{3s}\sqrt{1+2\delta_{3s}}}<\mu<1$; when $\mu>1$, $\rho<1$ is equivalent to $1<
 \mu<\dfrac{1}{1+\delta_{3s}}+\dfrac{1-\delta_{3s}}{2\delta_{3s}\sqrt{1+2\delta_{3s}^2}}$; when $\mu=1$, $\rho<1$ is equivalent to $\delta_{3s}<0.5340$. Thus
 we finish the proof.

\subsection{\label{proof2}Proof of Theorem \ref{thm:iter-num}}
On the one hand, considering the similarity with HTP, our proof for the number of iterations of STP mainly follows the proof of \cite[corollary 3.6]{foucart2011hard}. According to steps 6 and 7 of STP, if $S^n=S$, then we will get the exact solution by solving the least squares problem in step 7. For all $i\in S$ and $j\in \overline{S}$, a sufficient condition to guarantee $S^n=S$ in step 6 is
\begin{eqnarray}
|(\u^n+\mu\bPhi^{\! *}(\y-\bPhi\u^n))_i|&>&\!\!\!\!||(\u^n+\mu\bPhi^{\! *}(\y-\bPhi\u^n))_j|\nonumber\\
&=&\!\!\!\!|((\mu\bPhi^{\! *}\bPhi-\I)(\x-\u^n))_j|. \label{eq:iter1}
\end{eqnarray}
We observe that
\begin{eqnarray}
&&|(\u^n+\mu\bPhi^{\! *}(\y-\bPhi\u^n))_{i}|\nonumber\\
&=&|(\u^n+\mu\bPhi^{\! *}(\y-\bPhi\u^n)-\x+\x)_{i}|\nonumber\\
&\ge&\xi-|((\mu\bPhi^{\! *}\bPhi-\I)(\x-\u^n))_{i}|. \label{eq:iter2}
\end{eqnarray}
Then, we show that
\begin{eqnarray}
&&|((\mu\bPhi^{\! *}\bPhi-\I)(\x-\u^n))_{i}|+|((\mu\bPhi^{\! *}\bPhi-\I)(\x-\u^n))_{j}| \nonumber\\
        &\le&\sqrt{2}((\mu\bPhi^{\! *}\bPhi-\I)(\x-\u^n))_{i\cup j} \nonumber\\
        &\le&\sqrt{2}(|\mu-1|+\mu\delta_{2s+2})\|\x-\u^n\|_2 \nonumber\\
        &\le&\sqrt{2}(|\mu-1|+\mu\delta_{3s})\|\x-\u^n\|_2 \nonumber\\
        &\le&\sqrt{2}(|\mu-1|+\mu\delta_{3s})\sqrt{\frac{2\delta_{3s}^2(1+2\delta_{3s}^2)}{1-\delta_{3s}^2}}\|\x-\x^{n-1}\|_2\nonumber\\
        &\le&\sqrt{2}(|\mu-1|+\mu\delta_{3s})\sqrt{\frac{2\delta_{3s}^2(1+2\delta_{3s}^2)}{1-\delta_{3s}^2}}\rho^{n-1}\|\x\|_2\nonumber\\
        &=&\sqrt{(1-\delta_{3s}^2)}\rho^n\|\x\|_2\nonumber\\
        &\le&\rho^n\|\x\|_2.
\end{eqnarray}
So \eqref{eq:iter1} is satisfied as soon as
\begin{eqnarray}
\xi\ge\rho^n\|\x\|_2.
\end{eqnarray}
Then the smallest $n$ is
\be
\left\lceil \frac{\ln\|\x\|_2/\xi}{\ln {1}/{\rho}} \right\rceil.\label{eq:iter7}
\ee

On the other hand, assuming that $x$ is  $s$-sparse and setting $\e^{\prime}=\0$, combing \eqref{eq:stp1} and \eqref{eq:stp2}, we have
\be
\|\x_{\overline{S^n}}\|_2\le2(|\mu-1|+\mu\delta_{3s})\delta_{3s}\sqrt{\frac{1+2\delta_{3s}^2}{1-\delta_{3s}^2}}\|\x-\x^{n-1}\|_2. \label{eq:iter3}
\ee
Setting $\e^{\prime}=\0$ and substituting $n-1$ for $n$ in \eqref{eq:stp3}, one has
\be
\|\x-\x^{n-1}\|_2\le\sqrt{\frac{1}{1-\delta_{2s}^2}}\|\x_{\overline{S^{n-1}}}\|_2.\label{eq:iter4}
\ee
Combing \eqref{eq:iter3} and \eqref{eq:iter4} and magnifying $\delta_{2s}$ to $\delta_{3s}$, we have
\be
\|\x_{\overline{S^n}}\|_2\le\rho\|\x_{\overline{S^{n-1}}}\|_2.\label{eq:iter5}
\ee
where $\rho$ referred to \eqref{eq:rho}.

\eqref{eq:iter5} has the same form with \cite[Theorem 2 (6)]{dai2009subspace}, but different geometry rate $\rho$, so \cite[Theorem 8]{dai2009subspace} is also suitable for STP by using the new $\rho$ in \eqref{eq:rho} instead of the corresponding one ``$c_K$'' in \cite{dai2009subspace}. So any $s$-sparse vector $\x\in\bbR^{N}$ is reconstructed by STP with $\y=\bPhi\x$ in at most
\be
\left\lceil \frac{1.5s}{\ln1/\rho}\right\rceil. \label{eq:iter6}
\ee
Combing \eqref{eq:iter7} and \eqref{eq:iter6}, then we get Theorem \ref{thm:iter-num}.
\bibliographystyle{IEEETran}
\bibliography{Bib}
\end{document}